\documentclass[11pt,a4paper]{article}
\usepackage{amsmath,amsthm,amsfonts,amssymb}

\usepackage{graphicx,color}
\usepackage{boxedminipage}
\newtheorem{theorem}{Theorem}

\newtheorem{corollary}{Corollary}
\newtheorem{lemma}{Lemma}

\newtheorem{define}{Definition}

\usepackage{todonotes}

\newcommand{\lv}[1]{}

\usepackage{vmargin}
\setmarginsrb{1in}{1in}{1in}{1in}{0mm}{0mm}{0mm}{7mm}

\newcommand{\tl}{{\mathbf{tl}}}
\newcommand{\mw}{{\mathbf{mw}}}

\newcommand{\dist}{\textrm{dist}}
\newcommand{\diam}{\textrm{diam}}
\newcommand{\md}{\textrm{md}}
\newcommand{\mdim}{{\sc Metric Dimension}}

\DeclareMathOperator{\operatorClassNP}{NP}
\newcommand{\classNP}{\ensuremath{\operatorClassNP}}

\DeclareMathOperator{\operatorClassFPT}{FPT}
\newcommand{\classFPT}{\ensuremath{\operatorClassFPT}}

\newcommand{\defproblemu}[3]{
  \vspace{1mm}
\noindent\fbox{
  \begin{minipage}{0.95\textwidth}
  #1 \\
  {\bf{Input:}} #2  \\
  {\bf{Question:}} #3
  \end{minipage}
  }
  \vspace{1mm}
}

\begin{document}

\title{Metric Dimension of Bounded Tree-length Graphs\thanks{A preliminary version of this paper appeared as an extended abstract in the proceedings of MFCS 2015.
Supported by the European Research Council (ERC) via grant Rigorous Theory of Preprocessing, reference 267959 and the the ELC project (Grant-in-Aid for Scientific Research on Innovative Areas, MEXT Japan).}}

\author{
R\'{e}my Belmonte\thanks{Department of Architecture and Architectural Engineering, Kyoto University, Japan. E-mail: {\tt{remybelmonte@gmail.com}}}
\and
Fedor V. Fomin\thanks{Department of Informatics, University of Bergen, Norway. E-mail: {\tt{\{fedor.fomin, petr.golovach\}@ii.uib.no}}}\addtocounter{footnote}{-1}
\and 
Petr A. Golovach\footnotemark {}
\and
M. S. Ramanujan\thanks{TU Wien, Vienna, Austria. E-mail: {\tt{msramanujan@gmail.com}}}}

\date{}

\maketitle

\begin{abstract}
The notion of \emph{resolving sets} in a graph was introduced by Slater (1975) and Harary and Melter (1976) as a way of uniquely identifying every vertex in a graph. A set of vertices in a graph is a resolving set if for any pair of vertices $x$ and $y$ there is a vertex in the set which has distinct distances to $x$ and $y$. A smallest resolving set in a graph is called a \emph{metric basis} and its size, the \emph{metric dimension} of the graph. The problem of computing the metric dimension of a graph is a well-known NP-hard problem and while it was known to be polynomial time solvable on trees, it is only recently that efforts have been made to understand its computational complexity on various restricted graph classes. In recent work, Foucaud et al. (2015)  showed that this problem is NP-complete even on interval graphs. They complemented this result by also showing that it is \emph{fixed-parameter tractable (FPT)} parameterized by the metric dimension of the graph. 
 In this work, we show that this FPT result can in fact be extended to all graphs of bounded tree-length. This includes well-known classes like chordal graphs, AT-free graphs and  permutation graphs. We also show that this problem is FPT parameterized by the modular-width of the input graph.
\end{abstract}

\section{Introduction} 
A vertex $v$ of a connected graph $G$ \emph{resolves} two distinct vertices $x$ and $y$ of $G$ if $\dist_G(v,x)\neq\dist_G(v,y)$, where $\dist_G(u,v)$ denotes the length of a shortest path between $u$ and $v$ in the graph $G$. A set of vertices $W\subseteq V(G)$ is a \emph{resolving} (or \emph{locating}) set for $G$ if for any two distinct $x,y\in V(G)$, there is $v\in V(G)$ that resolves $x$ and $y$. The \emph{metric dimension} $\md(G)$ is the minimum cardinality of a resolving set for $G$. This notion was introduced independently by Slater~\cite{Slater75} and Harary and Melter~\cite{HararyM76}. 
The task of the \textsc{Minimum Metric Dimension} problem is to find the metric dimension of a graph $G$. Respectively,

\defproblemu{\textsc{Metric Dimension}}{A connected graph $G$ and a positive integer $k$.}
{Is $\md(G)\leq k$?}

\noindent 
is the decision version of the problem.

The problem was first mentioned in the literature by Garey and Johnson \cite{GareyJ79} and the same authors later proved it to be NP-complete in general. Khuller et al.~\cite{KhullerRR96} have also shown that this problem is NP-complete on general graphs while more recently Diaz et al. \cite{DiazPSL12} showed that the
problem is NP-complete even when restricted to planar graphs. 
In this work, Diaz et al. also showed that this problem is solvable in polynomial time on the class of outer-planar graphs. Prior to this, not much was known about the computational complexity of this problem except that it is polynomial-time solvable on trees (see \cite{Slater75,KhullerRR96}), although  there are also results proving combinatorial bounds on the metric dimension of various graph classes \cite{ChartrandEJO00}.
 Subsequently, Epstein et al. \cite{EpsteinLW12} showed that this problem is NP-complete on split graphs, bipartite and co-bipartite graphs. They also showed that the \emph{weighted version} of {\mdim} can be solved in polynomial time on paths, trees, cycles, co-graphs and trees augmented with $k$-edges for fixed $k$. Hoffmann and Wanke \cite{HoffmannW12} extended the tractability results to a subclass of unit disk graphs and most recently, Foucaud et al. \cite{FoucaudMNPV14} showed that this problem is NP-complete on interval graphs.

The NP-hardness of the problem in general as well as on several special graph classes raises the natural question of resolving its parameterized complexity. Parameterized complexity is a two dimensional framework
for studying the computational complexity of a problem. One dimension is the input size
$n$ and another one is a parameter $k$. It is said that a problem is \emph{fixed parameter tractable} (or \classFPT), if it can be solved in time $f(k)\cdot n^{O(1)}$ for some function $f$.
We refer to the books of  Cygan et al.~\cite{CyganFKLMPPS15} and Downey and Fellows~\cite{DowneyF13}
for  detailed introductions  to parameterized complexity. The parameterized complexity of {\mdim} under the standard parameterization -- the metric dimension of the input graph, on general graphs was open until 2012, when Hartung and Nichterlein \cite{HartungN13} proved that it is W[2]-hard. The next natural step in understanding the parameterized complexity of this problem is the identification of special graph classes which permit fixed-parameter tractable algorithms (FPT). Recently, Foucaud et al. \cite{FoucaudMNPV14} showed that when the input is restricted to the class of interval graphs, there is an FPT algorithm for this problem parameterized by the metric dimension of the graph. However, as Foucaud et al. note, it is far from obvious how the crucial lemmas used in their algorithm for interval graphs might extend to natural superclasses like chordal graphs and charting the actual boundaries of tractability of this problem  remains an interesting open problem.

In this paper, we identify two width-measures of graphs, namely  \emph{tree-length} and \emph{modular-width} as two parameters under which we can obtain FPT algorithms for {\mdim}.  
The notion of tree-length was introduced by Dourisboure and Gavoille~\cite{DourisboureG07} in order to deal with tree-decompositions whose quality is measured not by the size of the bags but the \emph{diameter} of the bags. Essentially, the \emph{length} of a tree-decomposition is the  maximum diameter of the bags in this tree-decomposition and the tree-length of a graph is the minimum length over all tree-decompositions. The class of bounded tree-length graphs is an extremely rich graph class as it contains several well-studied graph classes like interval graphs, chordal graphs, AT-free graphs, permutation graphs and so on. As mentioned earlier, out of these, only interval graphs were known to permit FPT algorithms for {\mdim}. This provides a strong motivation for studying the role played by the tree-length of a graph in the computation of its metric dimension.
Due to the obvious generality of this class, our results for {\mdim} on this graph class significantly expand the  known boundaries of tractability of this problem (see Figure \ref{fig:classes}).
\begin{figure}[t]
	
\begin{center}

\includegraphics[height=170 pt, width=200 pt]{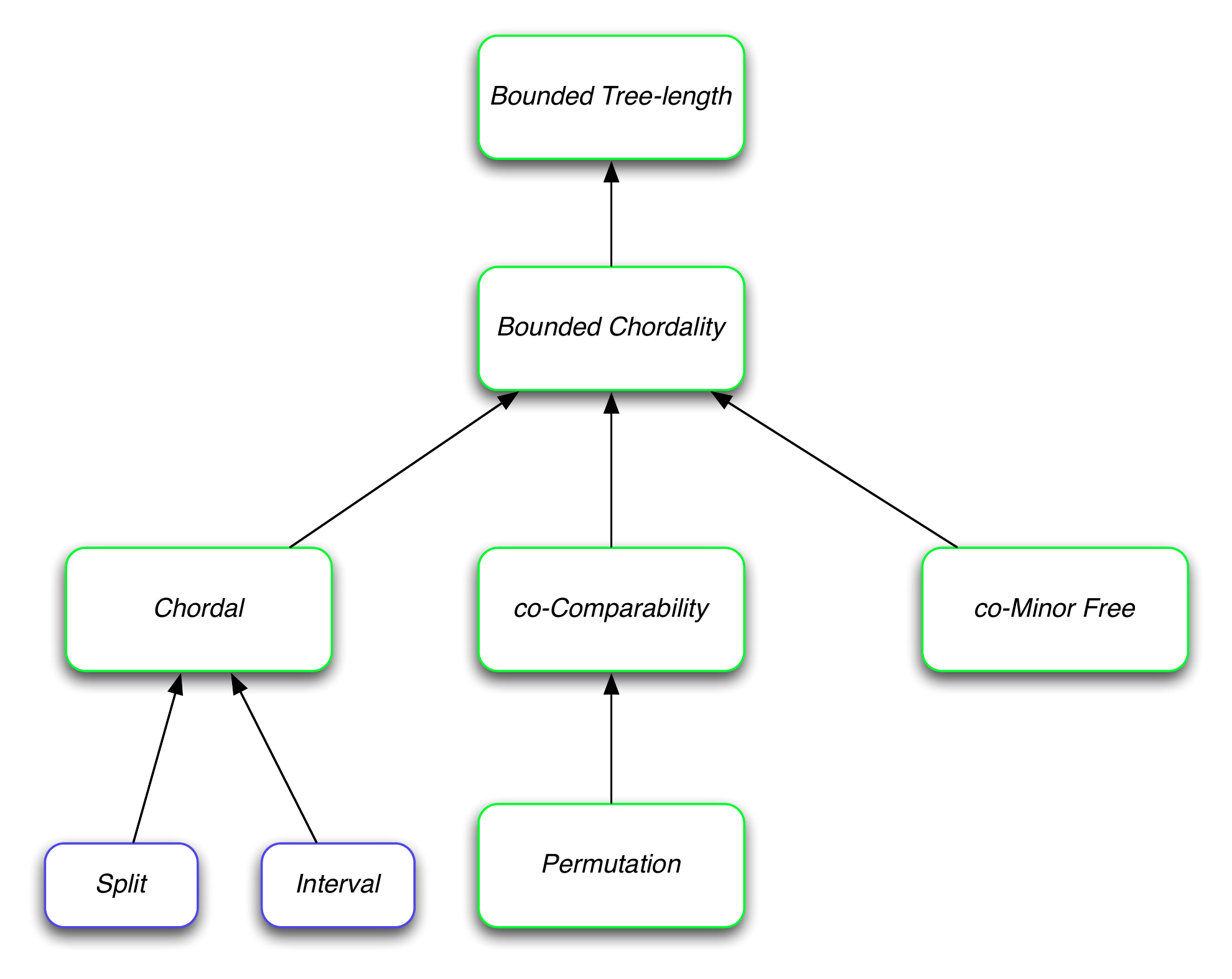}
\label{fig:classes}
\caption{Some well known graph classes which are subclasses of the class of bounded tree-length graphs. Out of these, {\mdim} was previously known to be FPT only on Split graphs and Interval graphs. Our results imply FPT algorithms parameterized by metric dimension on all other graph classes in the figure.}
\end{center}
\end{figure}
Modular-width was introduced by Gallai~\cite{Gallai67} in the context of comparability graphs and transitive orientations. A module in a graph is a set $X$ of vertices such that each vertex in $V\setminus X$ adjacent all or none of $X$. A partition of the vertex set into modules defines a \emph{quotient graph} with the set of modules as the vertex set. Roughly speaking, the modular decomposition tree is a rooted tree that represents the graph by recursively combining modules and quotient graphs. The modular-width of the decomposition is the size of the largest \emph{prime} node in this decomposition, that is, a node which cannot be partitioned into a set of non-trivial modules. Modular-width is a larger parameter than the more general clique-width and has been used in the past as a parameterization for problems where choosing clique-width as a parameter leads to W-hardness \cite{GajarskyLO13}.

Our main result is an FPT algorithm for {\mdim} parameterized by the \emph{maximum degree} and the tree-length of the input graph.

\begin{theorem}\label{thm:tl}
\textsc{Metric Dimension} is \classFPT{} when parameterized by $\Delta+\tl$, where $\Delta$ is the max-degree and $\tl$ is the tree-length of the input graph. 
\end{theorem}

It follows from (Theorem 3.6, \cite{KhullerRR96}) that for any graph $G$, $\Delta(G)\leq 2^{\md(G)}+\md(G) -1$. Therefore, one of the main consequences of this theorem is the following.

\begin{corollary}
	\textsc{Metric Dimension} is \classFPT{} when parameterized by $\tl+k$, where $k$ is the metric dimension of the input graph. 
\end{corollary}

  Further, it is known that chordal graphs and permutation graphs have tree-length at most 1 and 2 respectively. This follows from the definition in the case of chordal graphs. In the case of permutation graphs it is known that their chordality is bounded by 4 (see for example \cite{ChandranLS05}) and by using the result of Gavoille et al. \cite{GavoilleKKPP01} for any $h$-chordal graph $G$, $\tl(G)\leq h/2$ and a tree decomposition of length at most $h/2$  can be constructed in polynomial time.
  Therefore, we obtain FPT algorithms for {\mdim} parameterized by the solution size 
 on chordal graphs and permutation graphs.
This answers a  problem posed by Foucaud et al. \cite{FoucaudMNPV14} who proved a similar result for the case of interval graphs.

The algorithm behind Theorem \ref{thm:tl} is  a  
dynamic programming algorithm on a bounded \emph{width} tree-decomposition. However, it is not sufficient to have bounded tree-width (indeed it is open whether {\mdim} is polynomial time solvable on graphs of treewidth 2). This is mainly due to the fact that pairs of vertices can be resolved by a vertex `far away' from them hence making the problem extremely non-local.
However, we use delicate distance based arguments using the tree-length and degree bound on the graph to show that most pairs are trivially resolved by any vertex that is sufficiently far away from the vertices in the pair and furthermore, the pairs that are not resolved in this way must be resolved `locally'. We then design a dynamic programming algorithm  incorporating these structural lemmas and show that it is in fact an FPT algorithm for {\mdim} parameterized by max-degree and tree-length.

Our second result is an FPT algorithm for {\mdim} parameterized by the modular-width of the input graph. 

\begin{theorem}\label{thm:mw}
 {\mdim} is \classFPT{} when parameterized by the modular-width of the input graph.
\end{theorem}

\addtocounter{theorem}{-1}

\section{Basic definitions and preliminaries}\label{sec:defs}
\noindent
{\bf Graphs.}
We consider finite undirected graphs without loops or multiple
edges. The vertex set of a graph $G$ is denoted by $V(G)$, 
the edge set by  $E(G)$. We typically use $n$ and $m$ to denote the number of vertices and edges respectively.
For a set of vertices $U\subseteq V(G)$,
$G[U]$ denotes the subgraph of $G$ induced by $U$.
and by $G-U$ we denote the graph obtained form $G$ by the removal of all the vertices of $U$, i.e., the subgraph of $G$ induced by $V(G)\setminus U$. 
A set of vertices $U\subset V(G)$ is a \emph{separator} of a connected graph $G$ if $G-U$ is disconnected. 
Let $G$ be a graph. 
For a vertex $v$, we denote by $N_G(v)$ its
\emph{(open) neighborhood}, that is, the set of vertices which are
adjacent to $v$. 
The \emph{distance} $\dist_G(u,v)$ between two vertices $u$ and $v$ in a connected graph $G$ is the number of edges in a shortest $(u,v)$-path. 
For a positive integer $r$, $N_G^r[v]=\{u\in V(G)\mid \dist_G(u,v)\leq r\}$.
For a vertex $v\in V(G)$ and a set $U\subseteq V(G)$, $\dist_G(v,U)=\min\{\dist_G(v,u)\mid u\in U\}$.
For a set of vertices $U\subseteq V(G)$, its \emph{diameter} $\diam_G(U)=\max\{\dist_G(u,v)\mid u,v\in U\}$. The \emph{diameter of a graph} $G$ is $\diam(G)=\diam_G(V(G))$. 
A vertex $v\in V(G)$ is \emph{universal} if $N_G(v)=V(G)\setminus\{v\}$.
For two graphs $G_1$ and $G_2$ with $V(G_1)\cap V(G_2)$, the \emph{disjoint union} of $G_1$ and $G_2$ is the graph with the vertex set $V(G_1)\cup V(G_2)$ and the edge set $E(G_1)\cup E(G_2)$, and the \emph{join} of $G_1$ and $G_2$ is the graph the vertex set $V(G_1)\cup V(G_2)$ and the edge set $E(G_1)\cup E(G_2)\cup\{uv\mid u\in V(G_1), v\in V(G_2)\}$.
For a positive integer $k$, a graph $G$ is \emph{$k$-chordal} if  the length of the longest induced cycle in $G$ is at most $k$. The 
\emph{chordality} of $G$  is the smallest integer $k$ such that $G$ is $k$-chordal. It is usually assumed that forests have chordality 2; \emph{chordal} graphs are 3-chordal graphs. 
We say that a set of vertices $W\subseteq V(G)$ \emph{resolves} a set of vertices $U\subseteq V(G)$ if for any two distinct vertices $x,y\in U$, there is a vertex $v\in W$ that resolves them. Clearly, $W$ is a resolving set for $G$ if $W$ resolves $V(G)$.

\medskip
\noindent
{\bf Modular-width.}
A set $X\subseteq V(G)$ is a \emph{module} of graph $G$ if for any $v\in V(G)\setminus X$, either 
$X\subseteq N_G(v)$ or $X\cap N_G(v)=\emptyset$. The \emph{modular-width} of a graph $G$ introduced by Gallai in~\cite{Gallai67} is 
is the maximum size of a prime node in the modular decomposition
tree. For us, it is more convenient to use the following recursive definition.
The modular-width of a graph $G$ is at most $t$ if one of the following holds:
\begin{itemize}
\item[i)] $G$ has one vertex, 
\item[ii)] $G$ is disjoint union of  two graphs of modular-width at most $t$,
\item[iii)] $G$ is a join of two graphs of modular-width at most $t$, 
\item[iv)] $V(G)$ can be partitioned into $s\leq t$ modules $X_1,\ldots,X_s$ such that $\mw(G[X_i])\leq t$ for $i\in\{1,\ldots,s\}$.
\end{itemize}
The modular-width of a graph can be computed in linear time by the algorithm of Tedder et al.~\cite{TedderCHP08} (see also~\cite{HabibP10}).
Moreover, this algorithm outputs the algebraic expression of $G$
corresponding to the described procedure of its construction.

\medskip
\noindent
{\bf Tree decompositions.} 
A \emph{tree decomposition} of a graph $G$ is a pair $(\mathcal{X},T)$ where $T$
is a tree and $\mathcal{X}=\{X_{i} \mid i\in V(T)\}$ is a collection of subsets (called {\em bags})
of $V(G)$ such that: 
\begin{enumerate}
\item $\bigcup_{i \in V(T)} X_{i} = V(G)$, 
\item for each edge $xy \in E(G)$, $x,y\in X_i$ for some  $i\in V(T)$, and 
\item for each $x\in V(G)$ the set $\{ i \mid x \in X_{i} \}$ induces a connected subtree of $T$.
\end{enumerate}
The \emph{width} of a tree decomposition $(\{ X_{i} \mid i \in V(T) \},T)$ is $\max_{i \in V(T)}|X_{i}| - 1$. 
The \emph{length} of a tree decomposition $(\{ X_{i} \mid i \in V(T) \},T)$ is $\max_{i \in V(T)}\diam_G(X_{i})$.
The \emph{tree-length} if  a graph $G$ denoted as $\tl(G)$ is the minimum length over all tree decompositions of $G$.

The notion of tree-length was introduced by Dourisboure and Gavoille~\cite{DourisboureG07}. Lokshtanov proved in~\cite{Lokshtanov10}  that
it is \classNP-complete to decide whether $\tl(G)\leq\ell$ for a given $G$ for any fixed $\ell\geq 2$, but it was shown by  Dourisboure and Gavoille in~\cite{DourisboureG07} that 
the tree-length can be approximated in polynomial time  within a factor of 3. 

We  say that a tree decomposition  $(\mathcal{X},T)$ of a graph $G$ with $\mathcal{X}=\{X_{i} \mid i\in V(T)\}$  is \emph{nice} if $T$ is a rooted binary tree such that  the nodes of $T$ are of four types:
\begin{enumerate}
 \item[i)] a \emph{leaf node} $i$ is a leaf of $T$ and $|X_i|=1$;
 \item[ii)] an \emph{introduce node} $i$ has one child $i'$ with $X_i=X_{i'}\cup\{v\}$ for some vertex $v\in V(G)\setminus X_{i'}$;
 \item[iii)] a \emph{forget node} $i$ has one child $i'$ with $X_i=X_{i'}\setminus\{v\}$ for some vertex $v\in X_{i'}$; and 
 \item[iv)] a \emph{join node} $i$ has two children $i'$ and $i''$  with $X_i=X_{i'}=X_{i''}$ such that the  subtrees of $T$ rooted in $i'$ and $i''$ have at least one forget vertex each.
\end{enumerate}  
By the same arguments as were used by Kloks in~\cite{Kloks94}, it can be proved that every tree decomposition of a graph can be converted in linear time to a nice tree decomposition of the same length and the same width $w$ such that  the size of the obtained tree is $O(wn)$. Moreover, for an arbitrary vertex $v\in V(G)$, it is possible to obtain such a nice tree decomposition with the property that $v$ is the unique vertex of the root bag.  

\section{{\mdim} on graphs of bounded tree-length and max-degree}\label{sec:tl}
In this section we prove that \textsc{Metric Dimension} is \classFPT{} when parameterized by the max-degree and tree-length of the input graph. Throughout the section we use the following notation. 
Let $(\mathcal{X},T)$, where $\mathcal{X}=\{X_i\mid i\in V(T)\}$, be a nice tree decomposition of a graph $G$. Then for $i\in V(T)$, $T_i$ is the subtree of $T$ rooted in $i$ and $G_i$ is the subgraph of $G$ induced by $\cup_{j\in V(T_i)}X_j$. We first begin with a subsection where we prove the required structural properties of graphs of bounded tree-length and max-degree.

\subsection{Properties of graphs of bounded tree-length and max-degree}

We need the following lemma from \cite{BodlaenderT97}, bounding the treewidth of graphs of bounded tree-length and degree.

\begin{lemma}{\rm{\cite{BodlaenderT97}}}\label{lem:width}
Let $G$ be a connected graph with $\Delta(G)=\Delta$ and let $(\mathcal{X},T)$ be a tree decomposition of $G$ with the length at most $\ell$.
Then the width of $(\mathcal{X},T)$, is at most $w(\Delta,\ell)=\Delta(\Delta-1)^{(\ell-1)}$.
\end{lemma}

We also need the next lemma which essentially bounds the number of bags of $(\mathcal{X},T)$ a particular vertex of the graph appears in. We then use this lemma to prove Lemma \ref{lem:tree-dist}, which states that the `distance between a pair of vertices in the tree-decomposition' in fact approximates the distance between these vertices in the graph by a factor depending only on $\Delta$ and $\ell$.

\begin{lemma}\label{lem:path}
Let $G$ be a connected graph with $\Delta(G)=\Delta$, and let $(\mathcal{X},T)$, where $\mathcal{X}=\{X_i\mid i\in V(T)\}$, be a nice tree decomposition of $G$ of length at most $\ell$.
Furthermore, let $P$ be a path in $T$ such that  for some vertex  $z\in V(G)$,  $z\in X_i$ for every $i\in V(P)$. Then $|V(P)|\leq \alpha(\Delta,\ell)=2(\Delta^\ell(\Delta+2)+4)$.
\end{lemma}

\begin{proof}
Let $P'$ be a path in $T$ such that $z\in X_i$ for $i\in V(P')$. Furthermore, suppose that one of the endpoints of $P'$ is an ancestor of the other endpoint in $T$. We will argue that $\vert V(P')\vert\leq \alpha(\Delta,\ell)/2$, which will in turn imply the lemma because for any path $P$ in $T$ such that $z\in X_i$ for $i\in V(P)$, there is a subpath $P'$ of length at least half that of $P$ where one of the endpoints is an ancestor of the other.
Now, denote by $n_j,n_i,n_f,n_l$ the number of join, introduce, forget and leaf nodes of $P'$. 

\begin{itemize}
	\item 
Denote by $I$ the set of introduce nodes of $P'$. 
Let $Z=\cup_{i\in I}X_i$. Observe that $Z\subseteq N^\ell_G(z)$. Therefore, $n_i\leq \vert Z\vert \leq \Delta^\ell$.

\item Denote by $F$ the set of forget nodes of $P'$, and let $Z=\cup_{i\in F}X_i$. Notice that $|Z|\geq n_f-2$ and  $Z\subseteq N^\ell_G(z)$. Therefore,  $n_f\leq |Z|+2\leq \Delta^\ell+2$.
\item Denote by $J$ the set of children of the join nodes of $P'$ that are outside $P'$. Notice that $|J|\geq n_j-1$.
Observe that for  $j\in J$, $T_j$ has at least one forget node. Therefore, for each $j\in J$, there is a vertex $x_j\in V(G_j)\setminus X_j$ adjacent to a vertex of $X_j$. 
Notice that the vertices $x_j$ for $j\in J$ are pairwise distinct and $\dist_G(z,x_j)\leq \ell+1$.  Consider $Z=\{x_j\mid j\in J\}\cup \{z\}$. We have that $Z\subseteq N^{\ell+1}_G(z)$ and $|Z|=|J|+1\geq n_j$. Therefore, $n_j\leq |Z|\leq \Delta^{\ell+1}$.
\end{itemize}
As $n_l\leq 2$, we obtain that $|V(P')|=n_j+n_i+n_f+n_l\leq \Delta^\ell(\Delta+2)+4$. 
\end{proof}

Using Lemma~\ref{lem:path}, we obtain the following. 

\begin{lemma}\label{lem:tree-dist}
Let $G$ be a connected graph with max-degree $\Delta(G)=\Delta$, and let $(\mathcal{X},T)$, where $\mathcal{X}=\{X_i\mid i\in V(T)\}$, be a nice tree decomposition of $G$ with the length at most $\ell$.
Then for $i,j\in V(T)$ and any $x\in X_i$ and $y\in X_j$, 
$$\dist_T(i,j)\leq \alpha(\Delta,\ell)(\dist_G(x,y)+1)-1.$$
\end{lemma}

 \begin{proof}
Consider $x\in X_i$ and $y\in X_j$ for $i,j\in V(T)$.
Let $R$ be a shortest $(x,y)$-path in $G$, and 
let $P$ be the unique $(i,j)$-path in $T$. Observe that for any $h\in V(P)$, $X_h$ contains at least one vertex of $R$.
Since any  vertex $z$ of $R$ is included in at most $\alpha(\Delta,\ell)$ bags $X_h$ for $h\in V(P)$ (Lemma \ref{lem:path}),
$|V(P)|\leq\alpha(\Delta,\ell)|V(R)|$ and, therefore, $\dist_T(i,j)\leq \alpha(\Delta,\ell)(\dist_G(x,y)+1)-1$.
\end{proof}

The following lemma is the main structural lemma based on which we design our algorithm. 

\begin{lemma}[Locality Lemma]\label{lem:resolv}
Let $(\mathcal{X},T)$, where $\mathcal{X}=\{X_i\mid i\in V(T)\}$, be a nice tree decomposition of length at most $\ell$ of a connected graph $G$ such that $T$ is rooted in $r$, $X_r=\{u\}$.
Let $\Delta=\Delta(G)$ be the max-degree of $G$ and let
$s= \alpha(\Delta,\ell)(2\ell+1)$. Then the following holds:
\begin{itemize}
\item[i)] If $i\in V(G)$ is an introduce node with the child $i'$ and $v$ is the unique vertex of $X_i\setminus X_{i'}$, then for any $x\in V(G_j)$ for a node $j\in V(T_{i})$ such that $\dist_T(i,j)\geq s$, $u$ resolves $v$ and $x$.
\item[ii)] If $i\in V(G)$ is a join node with the children $i',i''$ and 
$x\in V(G_{j})\setminus X_{j}$ for $j\in T_{i\rq{}}$  such that $\dist_T(i\rq{},j)\geq s-1$ and $y\in V(G_{i''})\setminus X_{i''}$,
 then $u$ or an arbitrary vertex $v\in (V(G_{j})\setminus X_{j})$ resolves $x$ and $y$.
\end{itemize}
\end{lemma}

\begin{proof}
To show i), consider  $x\in V(G_j)$  for some $j\in V(T_{i'})$ such that $\dist_T(i',j)\geq s$. 
As either $u\in X_i$ or $X_i$ separates $u$ and $x$,
$$\dist_G(u,x)=\min\{\dist_G(u,y)+\dist_G(y,z)+\dist_G(z,x)\mid y\in X_i,z\in X_{j}\}.$$
Let $y\in X_i$ and $z\in X_j$ be vertices such that $\dist_G(u,x)=\dist_G(u,y)+\dist_G(y,z)+\dist_G(z,x)$. Then by Lemma~\ref{lem:tree-dist},
$$\dist_G(u,x)\geq \dist_G(u,y)+\dist_G(y,z)\geq\dist_G(u,y)+\frac{s+1}{\alpha(\Delta,\ell)}-1.$$
Because $v\in X_i$ and $\diam_G(X_i)\leq \ell$, 
$$\dist_G(u,v)\leq \dist_G(u,y)+\dist_G(y,v)\leq \dist_G(u,y)+\ell.$$
Because $s= \alpha(\Delta,\ell)(2\ell+1)$, we obtain that $\dist_G(u,v)<\dist_G(u,x)$, completing the proof of the first statement.

To prove ii), let $x\in V(G_{j})$ for $j\in T_{i'}$ such that $\dist_T(i',j)\geq s-1$, and let $y\in V(G_{i''})\setminus X_{i''}$.
Assume also that  $v\in V(G_{j})\setminus X_{j}$.
Suppose that $u$ does not resolve $x$ and $y$. It means that $\dist_G(u,x)=\dist_G(u,y)$.
Because either $u\in X_i$ or $X_i$ separates $u$ and $\{x,y\}$,
there are $x',y'\in X_i$ such that
$\dist_G(u,x)=\dist_G(u,x')+\dist_G(x',x)$
and 
$\dist_G(u,y)=\dist_G(u,y')+\dist_G(y',y)$.
As $\dist_G(u,x)=\dist_G(u,y)$ and  $\diam_G(X_i)\leq \ell$,
$$\dist_G(x',x)-\dist_G(y',y)= \dist_G(u,y')-\dist_G(u,x')\leq \ell.$$
Notice that $\dist_G(x,X_i)\leq \dist_G(x,x')$ and $\dist_G(y,X_i)\geq \dist_G(y,y')-\ell$, because 
$\diam_G(X_i)\leq \ell$. Hence,
$\dist_G(x,X_i)-\dist_G(y,X_i)\leq 2\ell$.
There are $z,z'\in X_{j}$ such that 
$\dist_G(x,X_i)=\dist_G(x,z)+\dist_G(z,X_i)$
and
$\dist_G(v,X_i)=\dist_G(v,z')+\dist_G(z',X_i)$.
Because $\diam_G(X_{j})\leq \ell$,
$\dist_G(v,z)\leq \dist_G(v,z')+\ell$ and
$\dist_G(z,X_i)\leq \dist_G(z',X_i)+\ell$.
Hence,
$$\dist_G(v,z)+\dist_G(z,X_i)\leq \dist_G(v,z')+\dist_G(z',X_i)+2\ell\leq\dist_G(v,X_i)+2\ell.$$
Since $X_i$ separates $v$ and $y$,
\begin{align*}
\dist_G(v,y)&\geq\dist_G(v,X_i)+\dist_G(y,X_i)\\
&\geq\dist_G(v,z)+\dist_G(z,X_i)-2\ell+\dist_G(y,X_i)\\
&\geq\dist_G(v,z)+\dist_G(z,X_i)-2\ell+\dist_G(x,X_i)-2\ell\\
&\geq \dist_G(v,z)+2\dist_G(z,X_i)+\dist_G(x,z)-4\ell \\  
\end{align*}
Clearly, 
$\dist_G(v,x)\leq\dist_G(x,z)+\dist_G(v,z)$.
Hence,
 \begin{align*}
\dist_G(v,y)-\dist_G(v,x)&\geq (\dist_G(v,z)+2\dist_G(z,X_i)+\dist_G(x,z)-4\ell)\\
&-(\dist_G(x,z)+\dist_G(v,z))\\
&\geq 2\dist_G(z,X_i)-4\ell.
\end{align*}
It remains to observe that 
$\dist_G(z,X_i)\geq \frac{s+1}{\alpha(\Delta,\ell)}-1>2\ell$,
and we obtain that 
$\dist_G(v,y)-\dist_G(v,x)>0$, i.e., $v$ resolve $x$ and $y$.
\end{proof}

Having proved the necessary structural properties of graphs with bounded tree-length and max-degree, we proceed to set up some notation which will help us formally present our algorithm for {\mdim} on such graphs. 
However, before we do so, we will give an informal description of the way we use the above lemma to design our algorithm.

Let $i$ be a node in the tree-decomposition (see Figure \ref{fig:tree-length}) and suppose that it is an introduce node where the vertex $v$ is introduced. The case when $i$ is a join node can be argued analogously by appropriate applications of the statements of Lemma \ref{lem:resolv}. Since any vertex outside $G_i$ has at most $\ell+1$ possible distances to the vertices of $X_i$, the resolution of any pair in $G_i$ by a vertex outside can be expressed in a `bounded' way. The same holds for a vertex in $G_i-X_i$ which resolves a pair in $G-V(G_i)$. The tricky part is when a vertex in $G_i$ resolves a pair with at least one vertex in $G_i$.
Now, consider pairs of vertices in $G$ which are necessarily resolved by a vertex of the solution in $G_i$. Let $a,b$ be such a pair. Now, for those pairs $a,b$ such that both are contained in $G_{i'}$, either $v$ resolves them or we may inductively assume that these resolutions have been handled during the computation for node $i'$. We now consider other possible pairs.
Now, if $a$ is $v$, then by Lemma \ref{lem:resolv}, if $b$ is in $V(G_j)$ for any $j$ which is at a distance at least $s$ from $i$, then this pair is trivially resolved by $u$. Therefore, any `interesting pair' containing $v$ is contained within a distance of $s$ from $X_i$ in the tree-decomposition induced on $G_i$. However, due to Lemma \ref{lem:path} and the fact that $G$ has bounded degree, the number of such vertices which form an interesting pair with $v$ is bounded by a function of $\Delta$ and $\ell$. Now, suppose that $a$ is in $V(G_i)$ and $b$ is a vertex in $V(G)\setminus V(G_i)$ and there is an introduce node on the path from $i$ to the root which introduces $b$. Then, if $a\in V(G_j)$ where $j$ is at a distance at least $s$ from $i$, then this pair is trivially resolved by $u$. By the same reasoning, if the bag containing $a$ is within a distance of $s$ from $i$ then the node where $b$ is introduced must be within a distance of $s$ from $i$. Otherwise this pair is again trivially resolved by $u$. Again, there are only a bounded number of such pairs. Finally, suppose that $a\in V(G_i)$ and $b$ is not introduced on a bag on the path from $i$ to the root. In this case, there is a join node, call it $l$, on the path from $i$ to the root with children $l'$ and $l''$ such that $l'$ lies on the path from $i$ to the root and $b$ is contained in $V(G_{l''})$. In this case, we can use statement $ii)$ of Lemma \ref{lem:resolv} to argue that if $a$ lies in $V(G_j)$ where $j$ is at a distance at least $s$ from $i$ then it lies at a distance at least $s$ from $l$ and hence either $u$ or a vertex in $G_j$ resolves this and in the latter case, \emph{any} arbitrary vertex achieves this. Therefore, we simply compute solutions corresponding to both cases. Otherwise, the bag containing $a$ lies at a distance at most $s$ from $i$. In this case, if $l$ is at a distance greater than $s$ from $i$ then the previous argument based on statement $ii)$ still holds. Therefore, it only remains to consider the case when $l$ is at a distance at most $s$ from $i$. However, in this case, due to Lemma \ref{lem:tree-dist}, if $u$ does not resolve this pair, it must be the case that even $b$ lies in a bag which is at a distance at most $s$ from $l$. Hence, the number of such pairs is also bounded and we conclude that at any node $i$ of the dynamic program, the number of interesting pairs we need to consider is bounded by a function of $\Delta$ and $\ell$ and hence we can perform a bottom up parse of the tree-decomposition and compute the appropriate solution values at each node.

\begin{figure}[t]
\begin{center}
\includegraphics[height=300 pt, width=200 pt]{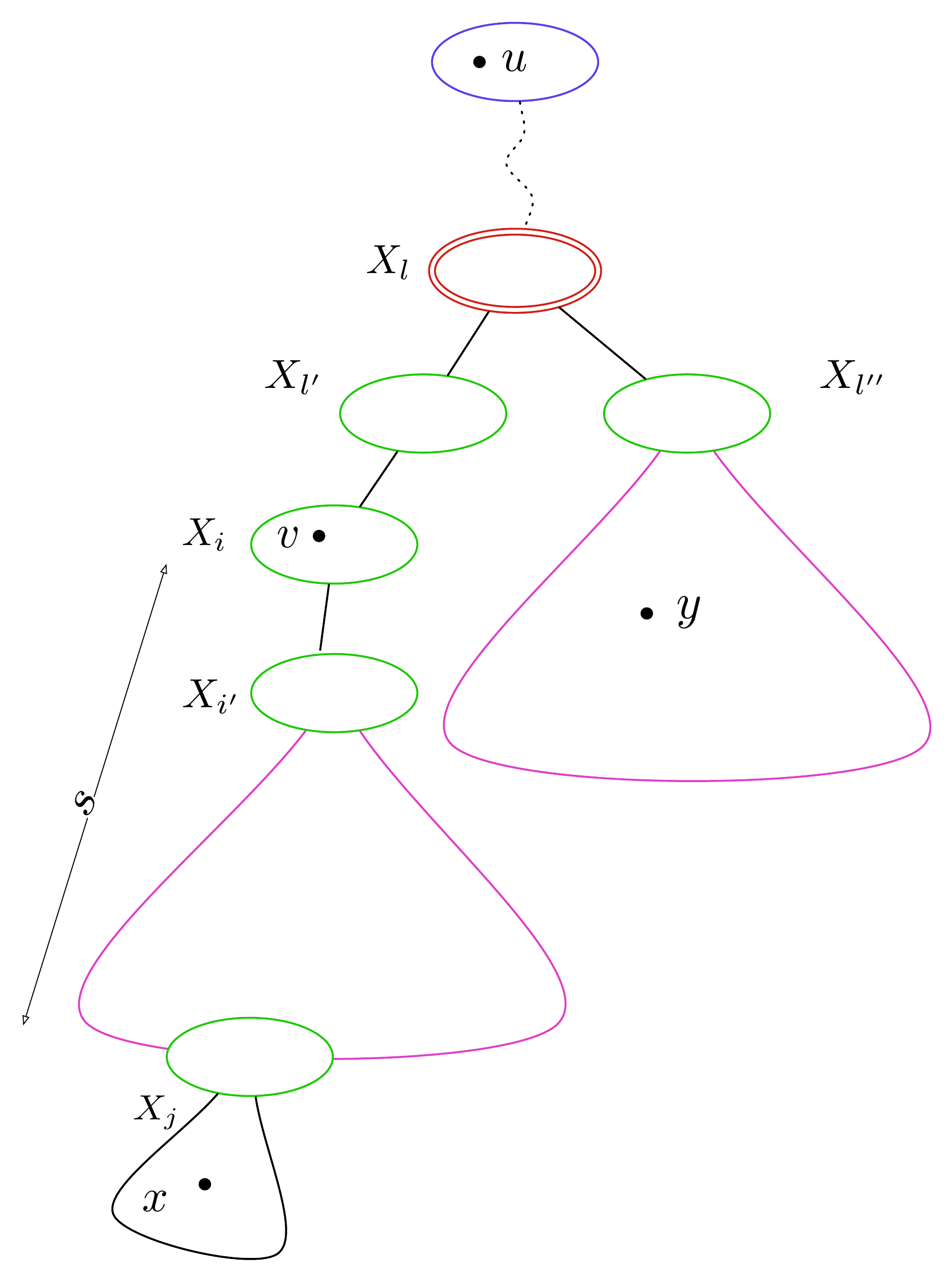}
\label{fig:tree-length}
\caption{An illustration of the structure guaranteed by Lemma \ref{lem:resolv}. Here, $X_i$ is an introduce node while $X_l$ is a join node. In this example, $l=i$ for statement $ii)$ of the lemma.}
\end{center}
\end{figure}

\subsection{Projections and resolving sets}

Let $X\subseteq V(G)$, and let $d$ be a positive integer such that $\diam_G(X)\leq d$. For a vertex $v\in V(G)$, we say that $\mathcal{P}r_{v,d}(X)=(X_0,\ldots,X_d)$, where 
$X_i=\{x\in X\mid \dist_G(v,x)=\dist_G(v,X)+i\}$ is the \emph{projection} of $v$ on $X$. Notice that $(X_0,\ldots,X_d)$ form an ordered partition  of $X$ (some sets could be empty), because $\diam_G(X)\leq d$. For a set $U\subseteq V(G)$, the set $\mathcal{P}r_{U,d}(X)=\{\mathcal{P}r_{v,d}(X)\mid v\in U\}$; notice that it can happen that $\mathcal{P}r_{v,d}(X)=\mathcal{P}_{u,d}(X)$
for $u,v\in U$, but as $\mathcal{P}r_{U,d}(X)$ is a set, it contains only one copy of $\mathcal{P}r_{v,d}(X)$.

Our algorithm uses the following properties of separators of bounded diameter. For the next two lemmas, let  $X$ be a separator of a connected graph $G$ such that $\diam_G(X)\leq d$, and let $V_1,V_2$ be a partition of the vertex set of $G-X$ such that no edge of $G$ joins a vertex of $V_1$ with a vertex of $V_2$.

\begin{lemma}\label{lem:sep}
If for $u,v\in V_1$, $\mathcal{P}r_{u,d}(X)=\mathcal{P}r_{v,d}(X)$, then $u$ resolves vertices $x,y\in V_2$ if and only if $v$ resolves $x,y$. Moreover, for a given ordered partition $(X_0,\ldots,X_d)$ of $X$, it can be decided in polynomial time whether a vertex $v\in V_1$ with $\mathcal{P}r_{v,d}(X)=(X_0,\ldots,X_d)$ resolves $x$ and $y$.  
\end{lemma}

\begin{proof}
Consider $v\in V_1$ and $x\in V_2$. Because $X$ separates $v$ and $x$,
\begin{align*}
\dist_G(v,x)&=\min\{\dist_G(v,x')+\dist_G(x',x)\mid x'\in X\}\\
&=\dist_G(v,X)+\min_{i\in\{0,\ldots,d\}}\min\{i+\dist_G(x',x)\mid x'\in X_i\}\\
&=\dist_G(v,X)+\min_{i\in\{0,\ldots,d\}}(i+\dist_G(X_i,x)).
\end{align*}
Therefore, $v\in V_1$ resolves $x,y\in V_2$ if and only if
$$\min_{i\in\{0,\ldots,d\}}(i+\dist_G(X_i,x))\neq\min_{i\in\{0,\ldots,d\}}(i+\dist_G(X_i,y)).$$
It immediately implies that if for $u,v\in V_1$, $\mathcal{P}r_{u,d}(X)=\mathcal{P}r_{v,d}(X)$, then $u$ resolves vertices $x,y\in V_2$ if and only if $v$ resolves $x,y$.
Because for any $x\in V_2$, $\min_{i\in\{0,\ldots,d\}}(i+\dist_G(X_i,x))$ can be computed in polynomial time by making use of the Dijkstra's algorithm if $(X_0,\ldots,X_d)$ is given, we obtain the second part of the statement. This completes the proof of the lemma.
\end{proof}

\begin{define}

Let $X'\subseteq X\cup V_2$ with $\diam_G(X)\leq d$. 
Let $(X_0,\ldots,X_d)$ be an ordered partition of $X$.
We define the ordered partition $(X_0',\ldots,X_d')$ of $X'$ as:
$$X_i'=\{x\in X'\mid \min_{i\in\{0,\ldots,d\}}(i+\dist_G(X_i,x))=s+i\}, where$$

$$s=\min_{x\in X'}\min_{i\in\{0,\ldots,d\}}(i+\dist_G(X_i,x)),$$
for $i\in\{0,\ldots,d\}$.
\end{define}

\begin{define}
We say that  $(X_0,\ldots,X_d)$ is a \textbf{$d$-cover} of $(X_0',\ldots,X_d')$ \emph{with respect to $V_1$}, 
and we say that $(X_0',\ldots,X_d')$ is \textbf{$d$-covered by} $(X_0,\ldots,X_d)$ \emph{with respect to $V_1$}.
We also say that a set $\mathcal{P}$ of ordered partitions $(X_0,\ldots,X_d)$ of $X$ is a \textbf{$d$-cover} of a set $\mathcal{P}'$ of ordered partition $(X_0',\ldots,X_d')$ of $X'$
\emph{with respect to $V_1$}, if   $\mathcal{P}'$ is the set of all ordered partitions of $X'$ that are $d$-covered by the partitions of $\mathcal{P}$.

\end{define}

Clearly, for a given $(X_0,\ldots,X_d)$, $(X_0',\ldots,X_d')$ can be constructed in polynomial time using, e.g., the Dijkstra's algorithm.

\begin{lemma}\label{lem:cover}
Let
$X'\subseteq X\cup V_2$ with $\diam_G(X)\leq d$. Let also $(X_0,\ldots,X_d)$ and $(X_0',\ldots,X_d')$ be ordered partitions of $X$ and $X'$ respectively such that 
$(X_0,\ldots,X_d)$ is a $d$-cover of $(X_0',\ldots,X_d')$ with respect to $V_1$. 
If $\mathcal{P}r_{v,d}(X)=(X_0,\ldots,X_d)$ for some $v\in V_1$, then $\mathcal{P}r_{v,d}(X')=(X_0',\ldots,X_d')$.
\end{lemma}

\begin{proof}
Let $v\in V_1$ and $x\in X'$. Suppose that $\mathcal{P}r_{v,d}(X)=(X_0,\ldots,X_d)$. Because $X$ separates $v$ and $x$,
$$\dist_G(v,x)=\dist_G(v,X)+\min_{i\in\{0,\ldots,d\}}(i+\dist_G(X_i,x)).$$
Hence,
$$\dist_G(v,X')=\dist_G(v,X)+\min_{x\in X'}\min_{i\in\{0,\ldots,d\}}(i+\dist_G(X_i,x)).$$
Let
$$s=\min_{x\in X'}\min_{i\in\{0,\ldots,d\}}(i+\dist_G(X_i,x))=\dist_G(v,X')-\dist_G(v,X).$$
Let $\mathcal{P}r_{v,d}(X')=(X_0',\ldots,X_d')$. 
We immediately obtain that
$$X_i'=\{x\in X'\mid \min_{i\in\{0,\ldots,d\}}(i+\dist_G(X_i,x))=s+i\}$$
for $i\in\{0,\ldots,d\}$, i.e., $(X_0,\ldots,X_d)$ is a $d$-cover of $(X_0',\ldots,X_d')$ with respect to $V_1$.
\end{proof}

\subsection{The algorithm}
Now we are ready to prove the main result of the section.

\addtocounter{theorem}{-1}
\begin{theorem}\label{thm:tl}
\textsc{Metric Dimension} is \classFPT{} when parameterized by $\Delta+\tl$, where $\Delta$ is the max-degree and $\tl$ is the tree-length of the input graph. 
\end{theorem}

\begin{proof}

Let $(G,k)$ be an instance of \textsc{Metric Dimension}. 
Recall that the tree-length of $G$ can be approximated in polynomial time within a factor of 3 by the results of Dourisboure and Gavoille~\cite{DourisboureG07}.
Hence, we assume that a tree-decomposition $(\mathcal{X},T)$ of length at most $\ell\leq 3\tl(G)+1$ is given. By Lemma~\ref{lem:width}, the width of  $(\mathcal{X},T)$ is at most $w(\Delta,\ell)$. 
We consider at most $n$ choices of a vertex $u\in V(G)$, and for each $u$, we check the existence of a resolving set $W$ of size at most $k$ that includes $u$. 

From now on, we assume that $u\in V(G)$ is given. We use the techniques of Kloks from~\cite{Kloks94} and construct from $(\mathcal{X},T)$ a nice tree decomposition of the same width and the length at most $\ell$ such that the root bag is $\{u\}$. To simplify notations, we assume that  $(\mathcal{X},T)$ is such a decomposition and $T$ is rooted in $r$. 
By Lemma~\ref{lem:path}, for any path $P$ in $T$,  any $z\in V(G)$ occurs in at most $\alpha(\Delta,\ell)$ bags $X_i$ for $i\in V(P)$.

We now design a dynamic programming algorithm over the tree decomposition that checks the existence of a resolving set of size at most $k$ that includes $u$. For simplicity, we only solve the decision problem. However, the algorithm can be modified to find such a resolving set (if exists). 

Let $s= \alpha(\Delta,\ell)(2\ell+1)$.
For $i\in V(T)$, we define $Y_i=\cup_{j\in N_{T_i}^s[i]}X_j$ and $I_i=\{j\in V(T_i)\mid \dist_{T_i}(i,j)=s\}$. Let also $I_i'=I_i\cup\{0\}$. 
  For each $i\in V(T)$, the algorithm constructs the table of values of the function $w_i(Z,\{\mathcal{P}^j\mid j\in I_i'\})$, where \vspace{10 pt}

\noindent\fbox{%
    \parbox{\textwidth}{%
\begin{itemize}
\item[i)] $Z\subseteq Y_i$ and $|Z|\leq k$,
\item[ii)] $\mathcal{P}^0$ is a set of ordered partitions $(Y_0,\ldots,Y_\ell)$ (some sets could be empty) of $X_i$ such that $\mathcal{P}r_{u,\ell}(X_i)\in \mathcal{P}^0$ if $u\notin X_i$,
% and $\mathcal{P}^0=\emptyset$ if $u\in X_i$,
\item[iii)] for $j\in I_i$, $\mathcal{P}^j$ is a set of ordered partitions $(Y_0,\ldots,Y_\ell)$ (some sets could be empty) of $X_j$,
\end{itemize}
and 
$w_i(Z,\{\mathcal{P}^j\mid j\in I_i'\})$ is the minimum cardinality of a set $W\subseteq V(G_i)$ such that
\begin{itemize}
\item[(a)] for any two distinct $x,y\in V(G_i)$, there is a vertex $v\in W$ that resolves $x$ and $y$ or there is an ordered partition $(Y_0,\ldots,Y_\ell)\in \mathcal{P}^0$ of $X_i$ such that 
a vertex $v\in V(G)\setminus V(G_i)$ with $\mathcal{P}r_{v,\ell}(X_i)=(Y_0,\ldots,Y_\ell)$ resolves $x$ and $y$,
\item[(b)] $W\cap Y_i=Z$,
\item[(c)] for $j\in I_i$, $\mathcal{P}^j=\mathcal{P}r_{W\cap (V(G_j)\setminus X_j),\ell}(X_j)$;
\end{itemize}
if such a set $W$ does not exist, then $w_i(Z,\{\mathcal{P}^j\mid j\in I_i'\})=+\infty$.
}}
\vspace{5 pt}

Notice that $G$ has a resolving set $W$ of size at most $k$ if and only if the table for the root node $r$ has an entry $w_r(Z,\{\mathcal{P}^j\mid j\in I_r'\})\leq k$.
Now we explain how we construct the table for each node $i\in V(T)$.

Let $i\in V(T)$. We define $J_i=\{j\in V(T_i)\mid \dist_{T_i}(i,j)=s-1\}$. For $Z$ and   $\{\mathcal{P}^j\mid j\in I_i\}$ satisfying i) and iii), 
$$\mathcal{R}(Z,\{\mathcal{P}^j\mid j\in I_i\})=\{\mathcal{R}^j\mid j\in J_i\},$$  where $\mathcal{R}^j$ is a set of ordered partitions $(Y_0,\ldots,Y_\ell)$ (some sets could be empty) of $X_j$, is constructed as follows. Let $j\in J_i$.
\begin{itemize}
\item If $j$ is a leaf node of $T$, then $\mathcal{R}^j=\emptyset$.
\item If $j$ is an introduce node of $T$ with the unique child $j'$, then 
$\mathcal{R}^j$ is the set of  ordered partitions  $(Y_0',\ldots,Y_\ell')$ of $X_j$
such that $\mathcal{P}^{j'}$ is an $\ell$-cover of $\mathcal{R}^j$ with respect to $V(G_{j'})\setminus X_{j'}$.
\item If $j$ is a forget node of $T$ with the unique child $j'$ and $\{v\}=X_{j'}\setminus X_j$, then we first construct  $\mathcal{R}^j$ as the set of  ordered partitions  $(Y_0',\ldots,Y_\ell')$ of $X_j$
such that $\mathcal{P}^{j'}$ is an $\ell$-cover of $\mathcal{R}^j$ with respect to $V(G_{j'})\setminus X_{j'}$, and then we 
set $\mathcal{R}^j=\mathcal{R}^j\cup\{\mathcal{P}r_{v,\ell}(X_i)\}$ if $v\in Z$.
\item If $j$ is a join node of $T$ with the two children $j'$ and $j''$, set $\mathcal{R}^j=\mathcal{P}^{j'}\cup \mathcal{P}^{j''}$.
\end{itemize}
Observe that given $Z$ and $\{\mathcal{P}^j\mid j\in I_i\}$, $\mathcal{R}(Z,\{\mathcal{P}^j\mid j\in I_i\})$ can be constructed in polynomial time.

\medskip
\noindent
{\bf Construction for a leaf node.} Let $X_i=\{x\}$. Then it is straightforward to verify that for any $\{\mathcal{P}^j\mid j\in I_i'\}$ satisfying ii) (notice that $I_i=\emptyset$), $w_i(\emptyset,\{\mathcal{P}^j\mid j\in I_i'\})=0$ and $w_i(\{x\},\{\mathcal{P}^j\mid j\in I_i'\})=1$.

\medskip
To describe the construction for introduce, forget and join nodes, assume that the tables are already constructed for the descendants of $i$ in $t$. We also initiate the construction by setting 
$w_i(\{x\},\{\mathcal{P}^j\mid j\in I_i'\})=+\infty$ for all $Z$ and $\{\mathcal{P}^j\mid j\in I_i'\}$ satisfying i)--iii). 

\medskip
\noindent
{\bf Construction for an introduce node.} Let $i'$ be the child of $i$ and $\{v\}=X_i\setminus X_{i'}$. Consider every $Z$ and $\{\mathcal{P}^j\mid j\in I_{i'}'\}$ satisfying i)--iii) for the node $i'$ such that $w_{i'}(Z,\{\mathcal{P}^j\mid j\in I_{i'}'\})<+\infty$. 

Notice that $J_{i'}=I_i$. We construct $\mathcal{R}(Z,\{\mathcal{P}^j\mid j\in I_{i'}\})=\{\mathcal{R}^j\mid j\in J_{i'}\}$ and for $j\in I_i$, set $\hat{\mathcal{P}}^j=\mathcal{R}^j$.
We consider two cases.

\medskip
\noindent
{\bf Case 1.} Set $\hat{Z}=Z\cap Y_i$ if $v\neq u$. We consider every set  $\hat{\mathcal{P}}^0$ of ordered partitions $(\hat{Y}_0,\ldots,\hat{Y}_\ell)$ of $X_i$ that satisfies ii) for the node $i$ such that $\hat{\mathcal{P}}^0$ is an $\ell$-cover of $\mathcal{P}^0$ with respect to $V(G)\setminus V(G_i)$. 

We verify the following condition:

\medskip
\noindent
{\bf Condition ($*$).} For every $x\in Y_{i}$
\begin{itemize}
\item there is $z\in \hat{Z}$ that resolves $x$ and $v$, or
\item there is an ordered partition $(Y_0,\ldots,Y_\ell)\in \hat{\mathcal{P}}^0$ of $X_i$ such that 
a vertex $z\in V(G)\setminus V(G_i)$ with $\mathcal{P}r_{z,\ell}(X_i)=(Y_0,\ldots,Y_\ell)$ resolves $x$ and $v$, or
\item there is an ordered partition $(Y_0,\ldots,Y_\ell)\in \mathcal{P}^h$ of $X_h$ for $h\in I_{i}'$ such that 
a vertex $z\in V(G_h)\setminus X_h$ with $\mathcal{P}r_{z,\ell}(X_h)=(Y_0,\ldots,Y_\ell)$ resolves $x$ and $v$.
\end{itemize}
Notice, that by Lemma~\ref{lem:sep}, ($*$) can be verified in polynomial time.
If ($*$) holds and $w_i(\hat{Z},\{\hat{\mathcal{P}}^j\mid j\in I_{i}'\})>w_{i'}(Z,\{\mathcal{P}^j\mid j\in I_{i'}'\})$, we set
$w_i(\hat{Z},\{\hat{\mathcal{P}}^j\mid j\in I_{i}'\})=w_{i'}(Z,\{\mathcal{P}^j\mid j\in I_{i'}'\})$.

\medskip
\noindent
{\bf Case 2.} Set $\hat{Z}=(Z\cap Y_i)\cup\{v\}$ if $|Z\cap Y_i|\leq k-1$. We consider every set  $\hat{\mathcal{P}}^0$ of ordered partitions $(\hat{Y}_0,\ldots,\hat{Y}_\ell)$ of $X_i$ that satisfies ii) for the node $i$ such that $\hat{\mathcal{P}}^0$ is an $\ell$-cover of $\mathcal{P}^0$ or $\mathcal{P}^0\setminus\{\mathcal{P}r_{v,\ell}(X_{i'})\}$
with respect to $V(G)\setminus V(G_i)$. 
If $w_i(\hat{Z},\{\hat{\mathcal{P}}^j\mid j\in I_{i}'\})>w_{i'}(Z,\{\mathcal{P}^j\mid j\in I_{i'}'\})+1$, we set
$w_i(\hat{Z},\{\hat{\mathcal{P}}^j\mid j\in I_{i}'\})=w_{i'}(Z,\{\mathcal{P}^j\mid j\in I_{i'}'\})+1$. Having described the way the algorithm computes the table at an introduce node, we now argue the correctness.

\medskip
\paragraph{{\bf Proof of correctness for an introduce node.}}
To show correctness, assume that  $w_i(\hat{Z},\{\hat{\mathcal{P}}^j\mid j\in I_{i}'\})$ is the value of $w_i$ obtained by the algorithm and denote by 
 $w_i^*(\hat{Z},\{\hat{\mathcal{P}}^j\mid j\in I_{i}'\})$ the value of the function by the definition, i.e., the the minimum cardinality of a set $W\subseteq V(G_i)$ satisfying iv)--vi).
We also assume inductively that the values of $w_{i'}$ are computed correctly.

We prove first that $w_i^*(\hat{Z},\{\hat{\mathcal{P}}^j\mid j\in I_{i}'\})\leq w_i(\hat{Z},\{\hat{\mathcal{P}}^j\mid j\in I_{i}'\})$ for $\hat{Z}$ and and  
$\{\hat{\mathcal{P}}^j\mid j\in I_{i}'\}$ satisfying i)--iii) for the node $i$.

If $w_i(\hat{Z},\{\hat{\mathcal{P}}^j\mid j\in I_{i}'\})=+\infty$, then the inequality holds trivially. Let $w_i(\hat{Z},\{\hat{\mathcal{P}}^j\mid j\in I_{i}'\})<+\infty$.
Then the value $w_i(\hat{Z},\{\hat{\mathcal{P}}^j\mid j\in I_{i}'\})$ is obtained as described above for some 
$Z$, $\{\mathcal{P}^j\mid j\in I_{i'}'\}$ satisfying i)--iii) for the node $i'$, and $\hat{\mathcal{P}}^0$ satisfying ii) for the node $i$. 
Clearly, $w_{i'}(Z,\{\mathcal{P}^j\mid j\in I_i'\})<+\infty$.
By induction,
$w_{i'}(Z,\{\mathcal{P}^j\mid j\in I_i'\})$ is the minimum cardinality of a set $W\subseteq V(G_{i'})$ satisfying iv)--vi) for the node $i'$.
Let $\hat{W}=W\cup \hat{Z}$.  

To show that iv) holds for $\hat{W}$,
consider  distinct $x,y\in V(G_i)$. 

If $x,y\in V(G_{i'})$, then there is a vertex $z\in W$ that resolves $x$ and $y$ or there is an ordered partition $(Y_0,\ldots,Y_\ell)$ of $X_{i'}$ such that 
a vertex $z\in V(G)\setminus V(G_{i'})$ with $\mathcal{P}r_{v,\ell}(X_{i'})=(Y_0,\ldots,Y_\ell)$ resolves $x$ and $y$. By Lemmas~\ref{lem:sep} and \ref{lem:cover}, if 
there is an ordered partition $(Y_0,\ldots,Y_\ell)$ of $X_{i'}$ such that a vertex $v\in V(G)\setminus V(G_{i'})$ with $\mathcal{P}r_{v,\ell}(X_{i'})=(Y_0,\ldots,Y_\ell)$ resolves $x$ and $y$, then there is an ordered partition $(\hat{Y}_0,\ldots,\hat{Y}_\ell)$ of $X_{i}$ that $\ell$-covers $(Y_0,\ldots,Y_\ell)$ with respect to $V(G)\setminus V(G_i)$ and 
we have that a vertex $v\in V(G)\setminus V(G_i)$ with $\mathcal{P}r_{v,\ell}(X_i)=(\hat{Y}_0,\ldots,\hat{Y}_\ell)$ resolves $x$ and $y$, or $v\in Z$ resolves $x$ and $y$ if $\mathcal{P}r_{v,\ell}(X_{i'})=(Y_0,\ldots,Y_l)$. 

Assume that $x=v$ and $y\in  V(G_{i'})$. If $v\in \hat{Z}$, then $v$ resolves $x$ and $y$. Suppose that $v\notin \hat{Z}$, i.e, the value of $w_i(\hat{Z},\{\hat{\mathcal{P}}^j\mid j\in I_{i}'\})$ was obtained in Case~1.
If $y\in V(G_j)\setminus X_j$ for $j\in I_{i}$, then $x$ and  $y$ are resolved by $u$ by Lemma~\ref{lem:resolv}. 
Let $y\in Y_i$. 
By ($*$),  there is $z\in \hat{Z}$ that resolves $y$ and $v$, or there is an ordered partition $(Y_0,\ldots,Y_\ell)\in \hat{\mathcal{P}}^0$ of $X_i$ such that 
a vertex $z\in V(G)\setminus V(G_i)$ with $\mathcal{P}r_{z,\ell}(X_i)=(Y_0,\ldots,Y_\ell)$ resolves $y$ and $v$, or
there is an ordered partition $(Y_0,\ldots,Y_\ell)\in \mathcal{P}_1^h$ of $X_h$ for $h\in I_{i'}$ such that 
a vertex $z\in V(G_h)\setminus X_h$ with $\mathcal{P}r_{z,\ell}(X_h)=(Y_0,\ldots,Y_\ell)$ resolves $y$ and $v$.
It remains to  observe that in the last case there is $z\in V(G_h)\setminus X_h$ with $\mathcal{P}r_{z,\ell}(X_h)=(Y_0,\ldots,Y_\ell)$ such that $z\in W\subseteq \hat{W}$, because v) holds  for $W$. 

Clearly,  $\hat{W}\cap Y_i=\hat{Z}$ by the definition, i.e., v) is fulfilled.

By the definition of $\mathcal{R}_{i'}$ and Lemma~\ref{lem:cover}, we obtain that for $j\in I_i$, $\hat{\mathcal{P}}^j=\mathcal{P}r_{\hat{W}\cap (V(G_j)\setminus X_j),\ell}(X_j)$ and vi) is satisfied.

Hence, $\hat{W}$ satisfies iv)--vi) for the node $i$ and, therefore, 
$w_i^*(\hat{Z},\{\hat{\mathcal{P}}^j\mid j\in I_{i}'\})\leq |\hat{W}|=w_i(\hat{Z},\{\hat{\mathcal{P}}^j\mid j\in I_{i}'\})$.

Now we prove that $w_i^*(\hat{Z},\{\hat{\mathcal{P}}^j\mid j\in I_{i}'\})\geq w_i(\hat{Z},\{\hat{\mathcal{P}}^j\mid j\in I_{i}'\})$.

If $w_i^*(\hat{Z},\{\hat{\mathcal{P}}^j\mid j\in I_{i}'\})=+\infty$, then the inequality holds. 
Assume that for 
$\hat{Z}$, $\{\hat{\mathcal{P}}^j\mid j\in I_{i'}'\}$ satisfying i)--iii) for the node $i$, $w_i^*(\hat{Z},\{\hat{\mathcal{P}}^j\mid j\in I_{i}'\})< +\infty$. Then there is 
$\hat{W}\subseteq V(G_i)$ satisfying iv)--vi) for the node $i$ and $w_i^*(\hat{Z},\{\hat{\mathcal{P}}^j\mid j\in I_{i}'\})= |\hat{W}|$.

Let $W=\hat{W}\cap V(G_{i'})$ and $Z=W\cap Z_{i'}$. We construct $\mathcal{P}^0$ as the set of ordered partitions of $X_{i'}$ such that  $\mathcal{P}^0$ is $\ell$-covered by $\hat{\mathcal{P}}^0$ and add $\mathcal{P}r_{v,\ell}(X_{i'})$ to this set if $v\in \hat{W}$. For  $j\in I_{i'}$, $\mathcal{P}^j=\mathcal{P}r_{W\cap (V(G_j)\setminus X_j),\ell}(X_j)$.
It is straightforward to see that $Z$ and $\{\mathcal{P}^j\mid j\in I_{i'}'\}$ satisfy i)--iii) for the node $i'$.
By the construction and Lemma~\ref{lem:cover}, $W$ satisfies iv)--vi) for the node $i'$ and the constructed $Z$ and $\{\mathcal{P}^j\mid j\in I_{i'}'\}$. 
Hence,  $w_{i'}(Z,\{\mathcal{P}^j\mid j\in I_{i'}'\})\leq |W|$. 

We claim that if $v\notin \hat{Z}$, then ($*$) is fulfilled. Because iv) is fulfilled for $\hat{W}$, for any $x\in Y_{i}$, there is a vertex $z\in W$ that resolves $v$ and $x$ or there is an ordered partition $(Y_0,\ldots,Y_\ell)\in \mathcal{P}^0$ of $X_i$ such that 
a vertex $x\in V(G)\setminus V(G_i)$ with $\mathcal{P}r_{v,\ell}(X_i)=(Y_0,\ldots,Y_\ell)$ resolves $v$ and $z$.
It is sufficient to notice that if $z\in W$ that resolves $v$ and $x$ and $z\notin \hat{Z}$, then 
$z\in  V(G_h)\setminus X_h$ for $h\in I_h'$ and, therefore, $\mathcal{P}r_{z,\ell}(X_h)\in \mathcal{P}^h$.

It remains to observe 
 that  the value of $w_i(\hat{Z},\{\hat{\mathcal{P}}^j\mid j\in I_{i}'\})$ constructed by the algorithm for 
$Z$, $\{\mathcal{P}^j\mid j\in I_{i'}'\}$ and $\hat{\mathcal{P}}^0$ is at most $|\hat{W}|=w_i^*(\hat{Z},\{\hat{\mathcal{P}}^j\mid j\in I_{i}'\})$.

\medskip
\noindent
{\bf Construction for a forget node.} 
Let $i'$ be the child of $i$ and $\{v\}=X_{i'}\setminus X_{i}$. Consider every $Z$ and $\{\mathcal{P}^j\mid j\in I_{i'}'\}$ satisfying i)--iii) for the node $i'$ such that $w_{i'}(Z,\{\mathcal{P}^j\mid j\in I_{i'}'\})<+\infty$. 
Recall that $J_{i'}=I_i$. We construct $\mathcal{R}(Z,\{\mathcal{P}^j\mid j\in I_{i'}\})=\{\mathcal{R}^j\mid j\in J_{i'}\}$ and for $j\in I_i$, set $\hat{\mathcal{P}}^j=\mathcal{R}^j$.
We set  $\hat{Z}=Z\cap Y_i$. We consider every set  $\hat{\mathcal{P}}^0$ of ordered partitions $(\hat{Y}_0,\ldots,\hat{Y}_\ell)$ of $X_i$ that satisfies ii) for the node $i$ such that $\hat{\mathcal{P}}^0$ is an $\ell$-cover of $\mathcal{P}^0$ with respect to $V(G)\setminus V(G_i)$. 
If $w_i(\hat{Z},\{\hat{\mathcal{P}}^j\mid j\in I_{i}'\})>w_{i'}(Z,\{\mathcal{P}^j\mid j\in I_{i'}'\})$, we set
$w_i(\hat{Z},\{\hat{\mathcal{P}}^j\mid j\in I_{i}'\})=w_{i'}(Z,\{\mathcal{P}^j\mid j\in I_{i'}'\})$.

\medskip
Correctness is proved in the same way as for the construction for an introduce node. Notice that the arguments, in fact, become simpler, because $V(G_i)\subseteq V(G_{i'})$.

\medskip
\noindent
{\bf Construction for a join node.} Let $i'$ and $i''$ be the children of $i$. Recall that $X_i=X_{i'}=X_{i''}$.
Consider every $Z_1$ and $\{\mathcal{P}_1^j\mid j\in I_{i'}'\}$ satisfying i)--iii) for the node $i'$ such that $w_{i'}(Z,\{\mathcal{P}_1^j\mid j\in I_{i'}'\})<+\infty$
and  every $Z_2$ and $\{\mathcal{P}_2^j\mid j\in I_{i''}'\}$ satisfying i)--iii) for the node $i''$ such that $w_{i''}(Z,\{\mathcal{P}_2^j\mid j\in I_{i'}'\})<+\infty$
with the property that 
$Z_1\cap X_i=Z_2\cap X_i$. 

We set $Z=(Z_1\cup Z_2)\cap Y_i$.

For every $j\in I_{i'}$, we construct the set $\mathcal{S}^j_1$
of ordered partitions $(Y_0,\ldots,Y_\ell)$ of $X_i$ such that $\mathcal{P}^j_1$ is an $\ell$-cover of $\mathcal{S}^j_1$, and set 
$$\mathcal{S}_1=(\cup_{j\in I_{i'}}\mathcal{S}^j_1)\cup(\cup_{v\in Z_1\setminus X_i}\mathcal{P}r_{v,\ell}(X_i)).$$
Similarly, for every $j\in I_{i''}$, we construct the set $\mathcal{S}^j_1$
of ordered partitions $(Y_0,\ldots,Y_\ell)$ of $X_i$ such that $\mathcal{P}^j_2$ is an $\ell$-cover of $\mathcal{S}^j_2$, and set 
$$\mathcal{S}_2=(\cup_{j\in I_{i''}}\mathcal{S}^j_2)\cup (\cup_{v\in Z_2\setminus X_i}\mathcal{P}r_{v,\ell}(X_i)).$$
We consider every set $\mathcal{P}^0$ of the ordered partitions $(Y_0,\ldots,Y_\ell)$ of $X_i$ that satisfy ii) for the node $i$ such that
$\mathcal{P}_1^0=\mathcal{P}^0\cup \mathcal{S}_2$ and 
$\mathcal{P}_2^0=\mathcal{P}^0\cup \mathcal{S}_1$.

Notice that $I_i=J_{i'}\cup J_{i''}$. We construct  $\mathcal{R}(Z_1,\{\mathcal{P}_1^j\mid j\in I_{i'}\})=\{\mathcal{R}^j\mid j\in J_{i'}\}$ and
 $\mathcal{R}(Z_2,\{\mathcal{P}_2^j\mid j\in I_{i''}\})=\{\mathcal{R}^j\mid j\in J_{i''}\}$. We define $\{\mathcal{P}^j\mid j\in I_i'\}$ by 
setting $\mathcal{P}^j=\mathcal{R}^j$ for $j\in J_{i'}\cup J_{i''}$.

We verify the following conditions:

\medskip
\noindent
{\bf Condition ($**$).} For every $x\in V(G_{i'})\setminus X_i$ and $y\in V(G_{i''})\setminus X_i$,
\begin{itemize}
\item there is $v\in Z$ that resolves $x$ and $y$, or
\item there is an ordered partition $(Y_0,\ldots,Y_\ell)\in \mathcal{P}^0$ of $X_i$ such that 
a vertex $v\in V(G)\setminus V(G_i)$ with $\mathcal{P}r_{v,\ell}(X_i)=(Y_0,\ldots,Y_\ell)$ resolves $x$ and $y$, or
\item 
 there is an ordered partition $(Y_0,\ldots,Y_\ell)\in \mathcal{P}^j$ of $X_j$ for $j\in I_{i}$ such that 
$x,y\notin V(G_j)\setminus X_j$
and
a vertex $v\in V(G_j)\setminus X_j$ with $\mathcal{P}r_{v,\ell}(X_j)=(Y_0,\ldots,Y_\ell)$ resolves $x$ and $y$, or
\item $x\in V(G_{j})\setminus X_j$ for $j\in I_{i}$ and  $\mathcal{P}^{j}\neq\emptyset$, or
\item $y\in V(G_j)\setminus X_j$ for $j\in I_{i}$ and  $\mathcal{P}^{j}\neq\emptyset$.
\end{itemize}
Notice, that by Lemma~\ref{lem:sep}, ($**$) can be verified in polynomial time.

If ($**$) holds and $w_i(Z,\{\mathcal{P}^j\mid j\in I_{i}'\})>w_{i'}(Z_1,\{\mathcal{P}_1^j\mid j\in I_{i'}'\})+w_{i''}(Z_2,\{\mathcal{P}_2^j\mid j\in I_{i'}''\})-|Z_1\cap X_i|$, 
we set
$w_i(Z,\{\mathcal{P}^j\mid j\in I_{i}'\})=w_{i'}(Z_1,\{\mathcal{P}_1^j\mid j\in I_{i'}'\})+w_{i''}(Z_2,\{\mathcal{P}_2^j\mid j\in I_{i'}''\})-|Z_1\cap X_i|$.

\medskip
\paragraph{{\bf Correctness for join nodes.}}
To show correctness, assume that  $w_i(Z,\{\mathcal{P}^j\mid j\in I_{i}'\})$ is the value of $w_i$ obtained by the algorithm and denote by 
 $w_i^*(Z,\{\mathcal{P}^j\mid j\in I_{i}'\})$ the value of the function by the definition, i.e., the the minimum cardinality of a set $W\subseteq V(G_i)$ satisfying iv)--vi).
We also assume inductively that the values of $w_{i'}$ and $w_{i''}$ are computed correctly.

We show first that $w_i^*(Z,\{\mathcal{P}^j\mid j\in I_{i}'\})\leq w_i(Z,\{\mathcal{P}^j\mid j\in I_{i}'\})$ for $Z$ and and  
$\{\mathcal{P}^j\mid j\in I_{i}'\}$ satisfying i)--iii) for the node $i$.

If $w_i(Z,\{\mathcal{P}^j\mid j\in I_{i}'\})=+\infty$, then the inequality trivially holds. Let $w_i(Z,\{\mathcal{P}^j\mid j\in I_{i}'\})<+\infty$.
Then the value $w_i(Z,\{\mathcal{P}^j\mid j\in I_{i}'\})$ is obtained as described above  for some 
$Z_1$, $\{\mathcal{P}_1^j\mid j\in I_{i'}'\}$ satisfying i)--iii) for the node $i'$, 
$Z_2$, $\{\mathcal{P}_2^j\mid j\in I_{i''}'\}$ satisfying i)--iii) for the node $i''$ and 
$\mathcal{P}^0$ satisfying ii) for the node $i$.
By induction,
$w_{i'}(Z_1,\{\mathcal{P}_1^j\mid j\in I_i'\})<+\infty$ is the minimum cardinality of a set $W_1\subseteq V(G_{i'})$ satisfying iv)--vi) for the node $i'$
and 
$w_{i''}(Z_2,\{\mathcal{P}_2^j\mid j\in I_i'\})<+\infty$ is the minimum cardinality of a set $W_2\subseteq V(G_{i''})$ satisfying iv)--vi) for the node $i''$.
Let $W=W_1\cup W_2$.  

To show that iv) holds for $W$, consider  distinct $x,y\in V(G_i)$. 

Suppose that $x,y\in V(G_{i'})$. Because iv) holds for $W_1$ and the node $i'$, there is a vertex $v\in W_1$ that resolves $x$ and $y$ or there is an ordered partition $(Y_0,\ldots,Y_\ell)\in \mathcal{P}_1^0$ of $X_{i'}$ such that 
a vertex $v\in V(G)\setminus V(G_{i'})$ with $\mathcal{P}r_{v,\ell}(X_i)=(Y_0,\ldots,Y_\ell)$ resolves $x$ and $y$. If there is a vertex $v\in W_1$ that resolves $x$ and $y$, then $v\in W$ resolve $x$ and $y$. Suppose that there is an ordered partition $(Y_0,\ldots,Y_\ell)\in \mathcal{P}_1^0$ of $X_{i'}$ such that 
a vertex $v\in V(G)\setminus V(G_{i'})$ with $\mathcal{P}r_{v,\ell}(X_i)=(Y_0,\ldots,Y_\ell)$ resolves $x$ and $y$. Recall that $\mathcal{P}_1^0=\mathcal{P}^0\cup \mathcal{S}_2$. If 
$(Y_0,\ldots,Y_\ell)\in \mathcal{P}^0$, then iv) holds. Let $(Y_0,\ldots,Y_\ell)\in \mathcal{S}_2$. Then there is $v\in Z_2\subseteq W_2\subseteq W$ such that $\mathcal{P}r_{v,\ell}(X_2)=(Y_0,\ldots,Y_\ell)$ and $v$ resolves $x$ and $y$, or there is  $S_2^j$ for $j\in I_{i''}$ such that 
$(Y_0,\ldots,Y_\ell)\in \mathcal{S}_2^j$. In the last case, there is a vertex $v\in W_2\cap (V(G_j)\setminus X_j)$ that resolves $x$ and $y$ by Lemmas~\ref{lem:sep} and \ref{lem:cover}. 

Clearly, the case $x,y\in V(G_{i''})$ is symmetric.

Assume that $x\in V(G_i)\setminus X_{i'}$ and $y\in V(G_{i''})\setminus X_{i''}$. Recall that ($**$) is fulfilled. If there is $v\in Z$ that resolves $x$ and $y$ or
there is an ordered partition $(Y_0,\ldots,Y_\ell)\in \mathcal{P}^0$ of $X_i$ such that 
a vertex $v\in V(G)\setminus V(G_i)$ with $\mathcal{P}r_{v,\ell}(X_i)=(Y_0,\ldots,Y_\ell)$ resolves $x$ and $y$, then $x$ and $y$ are resolved by $v\in Z\subseteq W$.
If %$x,y\in Y_i$ and
 there is an ordered partition $(Y_0,\ldots,Y_\ell)\in \mathcal{P}^j$ of $X_j$ for $j\in I_{i}$ such that 
$x,y\notin V(G_j)\setminus X_j$ and 
a vertex $v\in V(G_j)\setminus X_j$ with $\mathcal{P}r_{v,\ell}(X_j)=(Y_0,\ldots,Y_\ell)$ resolves $x$ and $y$, then there is such $v\in W$ and we again obtain that $x$ and $y$ are resolved by a vertex of $W$.
Suppose that the first three conditions of ($**$) are not fulfilled for $x$ and $y$. Then  $x\in V(G_{j})\setminus X_j$ for $j\in I_{i}$ and  $\mathcal{P}^{j}\neq\emptyset$ or $y\in V(G_j)\setminus X_j$ for $j\in I_{i}$ and  $\mathcal{P}^{j}\neq\emptyset$. 
If  $x\in V(G_{j})\setminus X_j$ for $j\in I_{i}$ and  $\mathcal{P}^{j}\neq\emptyset$, then there is $v\in W$ such that $v\in V(G_j)\setminus X_j$. By Lemma~\ref{lem:resolv} $u$ or $v$ resolve $x$ and $y$. Then case $y\in V(G_j)\setminus X_j$ for $j\in I_{i}$ and  $\mathcal{P}^{j}\neq\emptyset$ is symmetric.

We have that $W\cap Y_i=Z$ by the definition, i.e., v) is fulfilled.

By the definition of $\mathcal{R}_{i'}$,  $\mathcal{R}_{i'}$ and Lemma~\ref{lem:cover}, we obtain that for $j\in I_i$, $\mathcal{P}^j=\mathcal{P}r_{\hat{W}\cap (V(G_j)\setminus X_j),\ell}(X_j)$ and vi) is satisfied. 

Hence, $W$ satisfies iv)--vi) for the node $i$ and, therefore, 
$w_i^*(\hat{Z},\{\hat{\mathcal{P}}^j\mid j\in I_{i}'\})\leq |W|=w_i(\hat{Z},\{\hat{\mathcal{P}}^j\mid j\in I_{i}'\})$.

Now we prove that $w_i^*(Z,\{\mathcal{P}^j\mid j\in I_{i}'\})\geq w_i(Z,\{\mathcal{P}^j\mid j\in I_{i}'\})$.

If $w_i^*(Z,\{\mathcal{P}^j\mid j\in I_{i}'\})=+\infty$, then the inequality holds. 
Assume that for 
$Z$, $\{\mathcal{P}^j\mid j\in I_{i'}'\}$ satisfying i)--iii) for the node $i$, $w_i^*(Z,\{\mathcal{P}^j\mid j\in I_{i}'\})<+\infty$. Then there is 
$W\subseteq V(G_i)$ satisfying iv)--vi) for the node $i$ and $w_i^*(Z,\{\mathcal{P}^j\mid j\in I_{i}'\})= |W|$. 

Let $W_1=W\cap V(G_{i'})$ and $W_2=W\cap V(G_{i''})$. We define
$\mathcal{P}_1^0=\mathcal{P}^0\cup \mathcal{P}r_{W\setminus V(G_{i'}),\ell}(X_i)$ and 
$\mathcal{P}_2^0=\mathcal{P}^0\cup \mathcal{P}r_{W\setminus V(G_{i''}),\ell}(X_i)$.
For  $j\in I_{i'}$, $\mathcal{P}_1^j=\mathcal{P}r_{W\cap (V(G_j)\setminus X_j),\ell}(X_j)$, and
for $j\in I_{i''}$, $\mathcal{P}_2^j=\mathcal{P}r_{W\cap (V(G_j)\setminus X_j),\ell}(X_j)$.
It is straightforward to see that $Z_1$, $\{\mathcal{P}_1^j\mid j\in I_{i'}'\}$  
and $Z_2$, $\{\mathcal{P}_2^j\mid j\in I_{i''}'\}$
satisfy i)--iii) for the nodes $i'$ and $i''$ respectively.

To prove that $W_1$ satisfies iv)--vi) for the node $i'$  and  the constructed $Z_1$ $\{\mathcal{P}_1^j\mid j\in I_{i'}'\}$, it is sufficient to verify iv), as v) and vi) are straightforward.
Let $x,y\in V(G_{i\rq{}})$. There is a vertex  $v\in W$ that resolves $x$ and $y$ or there is an ordered partition $(Y_0,\ldots,Y_\ell)\in \mathcal{P}^0$ of $X_i$ such that 
a vertex $v\in V(G)\setminus V(G_i)$ with $\mathcal{P}r_{v,\ell}(X_i)=(Y_0,\ldots,Y_\ell)$ resolves $x$ and $y$. If there is $v\in W_1$ that resolves $x$ and $y$ or  
there is an ordered partition $(Y_0,\ldots,Y_\ell)\in \mathcal{P}^0\subseteq\mathcal{P}_1^0$ of $X_i$ such that 
a vertex $v\in V(G)\setminus V(G_i)$ with $\mathcal{P}r_{v,\ell}(X_i)=(Y_0,\ldots,Y_\ell)$ resolves $x$ and $y$, then we obtain iv) for $x$ and $y$. Assume that there is $v\in W\setminus W_1$ that resolves $x$ and $y$. Then $\mathcal{P}r_{v,\ell}(X_i)\in \mathcal{P}_1^0$ and we have that there is  an ordered partition $(Y_0,\ldots,Y_\ell)\in \mathcal{P}_1^0$ of $X_i$ such that 
a vertex $v\in V(G)\setminus V(G_i)$ with $\mathcal{P}r_{v,\ell}(X_i)=(Y_0,\ldots,Y_\ell)$ resolves $x$ and $y$.

We obtain that $W_1$ satisfies iv)--vi) for the node $i'$  and  the constructed $Z_1$ $\{\mathcal{P}_1^j\mid j\in I_{i'}'\}$ and, by the same arguments, 
$W_2$ satisfies iv)--vi) for the node $i''$  and  the constructed $Z_2$ $\{\mathcal{P}_2^j\mid j\in I_{i''}'\}$. 
Hence,  $w_{i'}(Z_1,\{\mathcal{P}_1^j\mid j\in I_{i'}'\})\leq |W_1|$ 
and $w_{i''}(Z_2,\{\mathcal{P}_2^j\mid j\in I_{i''}'\})\leq |W_2|$. 

Now we show that ($**$) is fulfilled. Let $x\in V(G_{i'})\setminus X_i$ and $y\in V(G_{i''})\setminus X_i$. 
Then there is $v\in W$ that resolves $x$ and $y$ or
there is an ordered partition $(Y_0,\ldots,Y_\ell)\in \mathcal{P}^0$ of $X_i$ such that 
a vertex $v\in V(G)\setminus V(G_i)$ with $\mathcal{P}r_{v,\ell}(X_i)=(Y_0,\ldots,Y_\ell)$ resolves $x$ and $y$. 
In the last case ($**$) holds for $x$ and $y$. Also we have the condition if $v\in Z$. Assume that $v\in W\setminus Z$.
Then $v\in V(G_j)\setminus X_j$ for some $j\in I_i$.
If $x,y\notin  V(G_j)\setminus X_j$, then we have the property that a vertex $v\in V(G_j)\setminus X_j$ with $\mathcal{P}r_{v,\ell}(X_j)=(Y_0,\ldots,Y_\ell)$ resolves $x$ and $y$.
Assume that $x\in V(G_j)\setminus X_j$ or $y\in V(G_j)\setminus X_j$. Then we have that $\mathcal{P}^{j}\neq\emptyset$ or
$\mathcal{P}^{j}\neq\emptyset$ respectively.
Therefore, ($**$) holds.

It remains to observe 
 that  the value of $w_i(Z,\{\mathcal{P}^j\mid j\in I_{i}'\})$ constructed by the algorithm for 
$Z_1$, $\{\mathcal{P}_1^j\mid j\in I_{i'}'\}$,  satisfying i)--iii) for the node $i'$, $Z_2$, $\{\mathcal{P}_2^j\mid j\in I_{i''}'\}$  and 
$\mathcal{P}^0$ 
 is at most $|W_1\cup W_2|=w_i^*(Z,\{\mathcal{P}^j\mid j\in I_{i}'\})$. 

It completes the correction proof for a join node and, therefore, we have that the algorithm correctly constructs the tables of values of $w_i(Z,\{\mathcal{P}^j\mid j\in I_{i}'\})$.

\paragraph{{\bf Running Time Analysis.}}
We now analyze the running time of the dynamic programming algorithm. For this,  we give the following upper bound on the size of each table. 
Let $i\in V(T)$. We have that $|X_i|\leq w(\Delta,\ell)$. We also have that $N_{T_i}^s\leq 2^{s+1}-1$. Hence, $|Y_i|\leq  (2^{s+1}-1)w(\Delta,\ell)$, and there is at most 
$2^{(2^{s+1}-1)w(\Delta,\ell)}$ possibilities to choose $Z$.
We have that $|I_i'|\leq 2^s+1$. The number of all ordered partitions $(Y_0,\ldots,Y_\ell)$ of any $X_j$ is at most $(\ell+1)^{|X_j|}\leq (\ell+1)^{ w(\Delta,\ell)}$. Hence, the table for the node $i$ contains at most   $2^{(2^{s+1}-1)w(\Delta,\ell)}(\ell+1)^{(2^s+1)w(\Delta,\ell)}$ values of the function $w_i(Z,\{\mathcal{P}^j\mid j\in I_{i}'\})$.

As the number of ordered partitions $(Y_0,\ldots,Y_\ell)$ of $X_i$ is at most $(\ell+1)^{w(\Delta,\ell)}$, we 
obtain that each table can be constructed in time $$O^*(2^{2(2^{s+1}-1)w(\Delta,\ell)}(\ell+1)^{(2^{s+1}+3)w(\Delta,\ell)}).$$ 
Then the total running time of the dynamic programming algorithm is $O^*(2^{2(2^{s+1}-1)w(\Delta,\ell)}(\ell+1)^{(2^{s+1}+3)w(\Delta,\ell)})$.

Since preliminary steps of our algorithm for \textsc{Metric Dimension} can be executed in polynomial time and we run the dynamic programming algorithm for at most $n$ choices of $u$, the total running time is $O^*(2^{2(2^{s+1}-1)w(\Delta,\ell)}(\ell+1)^{(2^{s+1}+3)w(\Delta,\ell)})$.
\end{proof}

\section{{\mdim} on graphs of bounded modular-width}\label{sec:mw}
In this section we prove that the metric dimension can be computed in linear time for graphs of bounded modular-width.
Let $X$ be a module of a graph $G$ and $v\in V(G)\setminus X$. Then the distances in $G$ between $v$ and the vertices of $X$ are the same. This observation immediately implies the following lemma.

\begin{lemma}\label{lem:mw}
Let $X\subset V(G)$ be a module of a connected graph $G$ and $|X|\geq 2$.  Let also $H$ be a graph obtained from $G[X]$ by the addition of a universal vertex.
Then any $v\in V(G)$ resolving $x,y\in X$ is a vertex of $X$, and  
if $W\subseteq V(G)$ is a resolving set of $G$, then $W\cap X$ resolves $X$ in $H$.
\end{lemma}

\begin{theorem}\label{thm:mw}
The metric dimension of a connected graph $G$ of modular-width at most $t$ can be computed in time $O(t^34^t\cdot n+m)$.
\end{theorem}

\begin{proof}
To compute $\md(G)$, we consider auxiliary values $w(H,p,q)$ defined for a (not necessarily connected) graph $H$ of modular-width at most $t$ with at least two vertices and boolean variables $p$ and $q$ as follows.
Let $H'$ be the graph obtained from $H$ by the addition of a universal vertex $u$. Notice that $\diam_{H'}(V(H))\leq 2$.
Then $w(H,p,q)$
the minimum size of a set $W\subseteq V(H)$ such that
\begin{itemize}
\item[i)] $W$ resolves $V(H)$ in $H'$,
\item[ii)] $H$ has a vertex $x$ such that $\dist_{H'}(x,v)=1$ for every $v\in W$ if and only if $p=true$, and 
\item[iii)] $H$ has a vertex $x$ such that $\dist_{H'}(x,v)= 2$ for every $v\in W$ if and only if $q=true$.
\end{itemize}
We assume that $w(H,p,q)=+\infty$ if such a set does not exists.
The intuition behind the definition of the function $w(.)$ is as follows. Let $X$ be a module in the graph $G$, $H=G[X]$ and let $H_1,\dots, H_s$ be the partition of $H$ into modules, of which $H_1,\dots, H_t$ are trivial. Let $Z$ be a hypothetical optimal resolving set and let $Z'=Z\cap X$. By Lemma \ref{lem:mw}, we know that every pair of vertices in $H$ must be resolved by a vertex in $Z'$. Therefore, we need to compute a set which, amongst satisfying other properties must be a resolving set for the vertices in $X$. However, since these vertices are all in the same module and $G$ is connected, any pair of vertices are either adjacent or at a distance exactly 2 in $G$. Hence, we ask for $W$ (condition (i)) to be a resolving set of $V(H)$ in $H'$, the graph obtained by adding a universal vertex to $H$. Further, it could be the case that a vertex $z$ in $Z'$ is required to resolve a pair of vertices, one contained in $X$ say $x$ and the other disjoint from $X$, say $y$. Now, if $x$ is at a distance 1 in $G$ (and hence $H'$) from every vertex in $Z'$ then for any vertex $x'\in X$ which is also at a distance exactly 1 from every vertex of $Z'$, $z$ is also required to resolve $x'$ and $y$. The same argument holds for vertices at distance exactly 2 from every vertex of $Z'$. Therefore, in order to keep track of such resolutions, it suffices to know \emph{whether} exists a vertex in $X$ which is at a distance exactly 1 (respectively 2) from every vertex of $Z'$. This is precisely what is captured by the boolean variables $p$ and $q$.

Recall that since $H$ has modular-width at most $t$, it
can be constructed from single vertex graphs by the disjoint union and join operation and decomposing $H$ into at most $t$ modules and $H$ has at least two vertices. In the rest of the proof, we 
we formally describe our algorithm to compute $w(H,p,q)$ given the modular decomposition of $H$ and the values computed for the `child' nodes. As the base case corresponds to graphs of size at most $t$ we may compute the values for the leaf nodes by brute force and execute a bottom up dynamic program.

\noindent
\paragraph{{\bf Description of the algorithm.}} We begin the description of the algorithm by first considering the cases when $H$ is the disjoint union or join of a pair of graphs. Following that, we consider the case when $H$ can be partitioned into at most $t$ graphs, each of modular-width at most $t$. Although the third case subsumes the first 2, we address these 2 cases explicitly for a clearer understanding of the algorithm.

\medskip

\noindent
{\bf Case~1.} $H$ is a disjoint union of $H_1$ and $H_2$. 
Assume without loss of generality that $|V(H_1)|\leq |V(H_2)|$. 

If $|V(H_1)|=|V(H_2)|=1$, then it is straightforward to verify that $w(H,false,true)=1$, $w(H,false,false)=2$ and $w(H,true,true)=w(H,true,false)=
+\infty$. 

Suppose that $|V(H_1)|=1$, $|V(H_2)|\geq 2$ and the values of $w(H_2,p,q)$ are already computed for $p,q\in\{true,false\}$. Clearly, the single vertex of $H_1$ is at distance 2 from any vertex of $H_2$ in $H'$. Observe that we have two possibilities of the vertex of $H_1$: it is either in a resolving set or not.
Then by Lemma~\ref{lem:mw},
\begin{itemize}
\item  $w(H,true,true)=w(H_2,true,false)$, 
\item $w(H,false,true)=\min\{w(H_2,false,false),w(H_2,true,true)+1,\linebreak w(H_2,false,true)+1\}$,
\item $w(H,true,false)=+\infty$,
\item $w(H,false,false)=\min\{w(H_2,true,false)+1,w(H_2,false,false)+1\}$.
\end{itemize}

Suppose that $|V(H_1)|,|V(H_2)|\geq 2$ and the values of $w(H_i,p,q)$ are already computed for $i\in\{1,2\}$ and $p,q\in\{true,false\}$. Notice that for $x\in V(H_1)$ and $y\in V(H_2)$, $\dist_{H'}(x,y)=2$. Observe also that any resolving set has at least one vertex in $H_1$ and at least one vertex in $H_2$.
Then by Lemma~\ref{lem:mw},
\begin{itemize}
\item  $w(H,true,true)=+\infty$, 
\item $w(H,false,true)=\min\{w(H_1,p_1,q_1)+w(H_2,p_2,q_2)\mid p_i,q_i\in\{true,false\}\linebreak\text{ for }i\in\{1,2\}\text{ and }q_1\neq q_2\}$,
\item $w(H,true,false)=+\infty$,
\item $w(H,false,false)=\min\{w(H_1,p_1,false)+w(H_2,p_2,false)\mid p_1,p_2\in\{true,false\}\}$.
\end{itemize}

\medskip
\noindent
{\bf Case~2.} $H$ is a join of $H_1$ and $H_2$. 
Assume without loss of generality that $|V(H_1)|\leq |V(H_2)|$. 

If $|V(H_1)|=|V(H_2)|=1$, then it is straightforward to verify that $w(H,true,false)=1$, $w(H,false,false)=2$ and $w(H,true,true)=w(H,false,true)=
+\infty$. 

Suppose that $|V(H_1)|=1$, $|V(H_2)|\geq 2$ and the values of $w(H_2,p,q)$ are already computed for $p,q\in\{true,false\}$. Clearly, the single vertex of $H_1$ is at distance 1 from any vertex of $H_2$ in $H'$, and this single vertex is in a resolving set or not.
Then by Lemma~\ref{lem:mw},
\begin{itemize}
\item  $w(H,true,true)=w(H_2,false,true)$, 
\item $w(H,false,true)=+\infty$,
\item $w(H,true,false)=\min\{w(H_2,false,false), w(H_2,true,true)+1,\linebreak w(H_2,true,false)+1\}$,
\item $w(H,false,false)=\min\{w(H_2,false,true)+1,w(H_2,false,false)+1\}$.
\end{itemize}

Suppose that $|V(H_1)|,|V(H_2)|\geq 2$ and the values of $w(H_i,p,q)$ are already computed for $i\in\{1,2\}$ and $p,q\in\{true,false\}$. Notice that for $x\in V(H_1)$ and $y\in V(H_2)$, $\dist_{H'}(x,y)=1$, and any resolving set has at least one vertex in $H_1$ and at least one vertex in $H_2$.
Then by Lemma~\ref{lem:mw},
\begin{itemize}
\item  $w(H,true,true)=+\infty$, 
\item $w(H,false,true)=+\infty$,
\item $w(H,true,false)=\min\{w(H_1,p_1,q_1)+w(H_1,p_2,q_2)\mid p_i,q_i\in\{true,false\}\linebreak\text{ for }i\in\{1,2\}\text{ and }p_1\neq p_2\}$,
\item $w(H,false,false)=\min\{w(H_1,false,q_1)+w(H_2,false,q_2)\mid q_1,q_2\in\{true,false\}\}$.
\end{itemize}

\begin{figure}[t]

\begin{center}

\includegraphics[height=144 pt, width=360 pt]{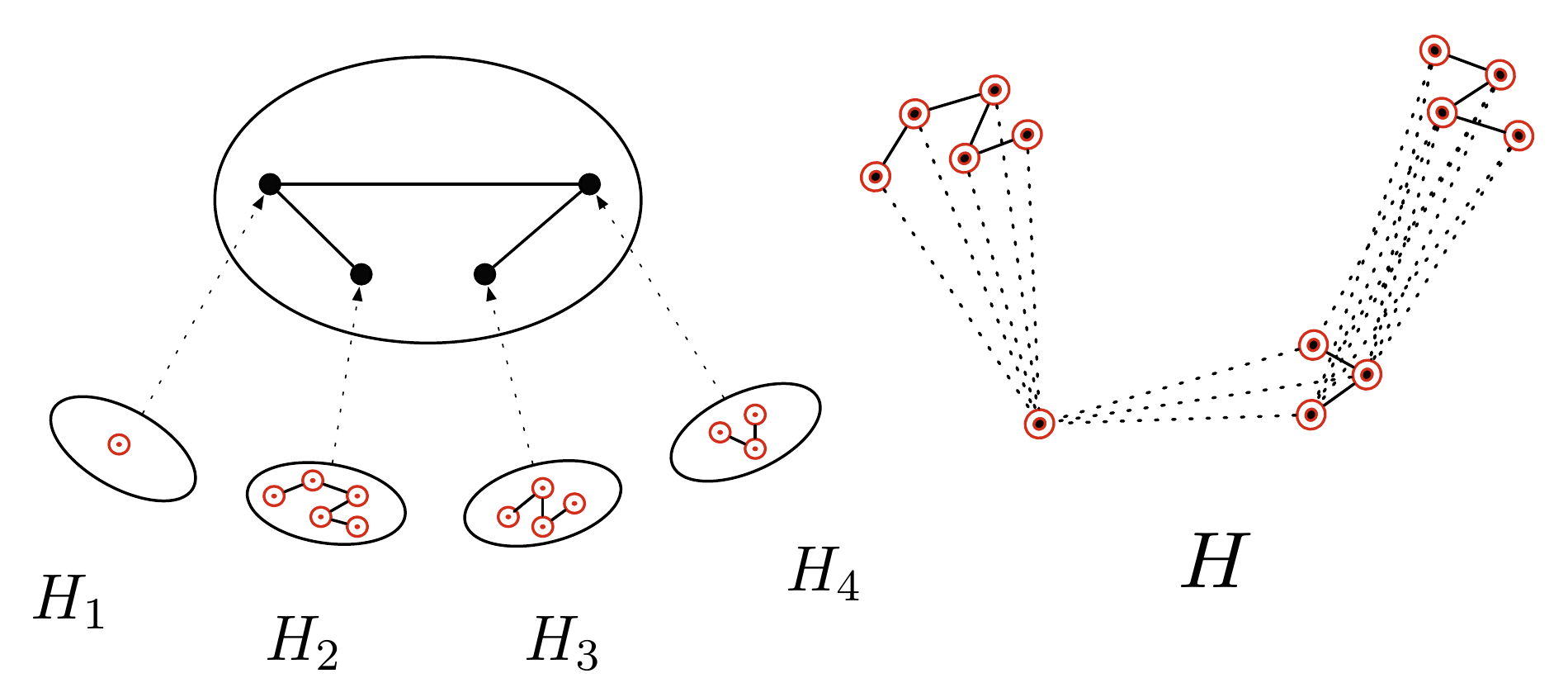}
\label{fig:modular-width}
\caption{An illustration of the modular decomposition of $H$ of width 5. Here, $H$ is partitioned into 4 modules $H_1,\dots, H_4$ where each module is a prime graph.}
\end{center}
\end{figure}

\medskip
\noindent
{\bf Case~3.} $V(H)$ is partitioned into $s\leq t$ non-empty modules $X_1,\ldots,X_s$, $s\geq 2$ (see for example Figure \ref{fig:modular-width}). 
Again we point out that, Cases~1 and 2 can be seen as special cases of Case~3, but we keep Cases~1 and 2 to make the algorithm for computing $w(H,p,q)$ more clear.
We assume that $X_1,\ldots,X_h$ are \emph{trivial}, i.e, $|X_i|=1$ for $i\in\{1,\ldots,h\}$; it can happen that $h=0$. 
Clearly,  for distinct $i,j\in\{1,\ldots,s\}$, either every vertex of $X_i$ is adjacent to every vertex of $X_j$ or the vertices of $X_i$ and $X_j$ are not adjacent. 

Consider the \emph{prime} graph $F$ with a vertex set $\{v_1,\ldots,v_s\}$ such that $v_i$ is adjacent to $v_j$ if and only if the vertices of $X_i$ are adjacent to the vertices of $X_j$ for distinct $i,j\in\{1,\ldots,s\}$. Let $F'$ be the graph obtained from $F$ by the addition of a universal vertex. Observe that if $x\in X_i$ and $y\in X_j$ for distinct $i,j\in\{1,\ldots,s\}$, then $\dist_{H'}(x,y)=\dist_{F'}(v_i,v_j)$.  

For boolean variables $p,q$, a set of indices $I\subseteq\{1,\ldots,h\}$ and boolean variables $p_i,q_i$, where $i\in\{h+1,\ldots,s\}$, we define 
$$\omega(p,q,I,p_{h+1},q_{h+2},\ldots,p_s,q_s)=|I|+\sum_{i=h+1}^s w(H[X_i],p_i,q_i)$$ 
if the following holds:
\begin{itemize}
\item[a)] the set $Z=\{v_i\mid i\in I\cup\{h+1,\ldots,s\}\}$ resolves $V(F)$ in $F'$,
\item[b)] if $p_i=true$ for some $i\in\{h+1,\ldots,s\}$, then for each $j\in\{1,\ldots,h\}\setminus I$, $\dist_{F'}(v_i,v_j)=2$ or there is $v_r\in Z$ such that $r\neq i,j$ and $\dist_{F'}(v_r,v_i)\neq \dist_{F'}(v_r,v_j)$,
\item[c)] if $q_i=true$ for some $i\in\{h+1,\ldots,s\}$, then for each $j\in\{1,\ldots,h\}\setminus I$, $\dist_{F'}(v_i,v_j)=1$ or there is $v_r\in Z$ such that $r\neq i,j$ and $\dist_{F'}(v_r,v_i)\neq \dist_{F'}(v_r,v_j)$,
\item[d)] if $p_i=p_j=true$ for some distinct $i,j\in \{h+1,\ldots,s\}$, then $\dist_{F'}(v_i,v_j)=2$ or there is $v_r\in Z$ such that $r\neq i,j$ and $\dist_{F'}(v_r,v_i)\neq \dist_{F'}(v_r,v_j)$,
\item[e)] if $q_i=q_j=true$ for some distinct $i,j\in \{h+1,\ldots,s\}$, then $\dist_{F'}(v_i,v_j)=1$ or there is $v_r\in Z$ such that $r\neq i,j$ and $\dist_{F'}(v_r,v_i)\neq \dist_{F'}(v_r,v_j)$,
\item[f)] $p=true$ if and only if there is $i\in\{1,\ldots,h\}\setminus I$ such that $\dist_{F'}(v_i,v_j)=1$ for $v_j\in Z$ or 
there is $i\in\{h+1,\ldots,s\}$ such that $p_i=true$ and $\dist_{F'}(v_i,v_j)=1$ for $v_j\in Z\setminus\{v_i\}$,
\item[g)] $q=true$ if and only if there is $i\in\{1,\ldots,h\}\setminus I$ such that $\dist_{F'}(v_i,v_j)=2$ for $v_j\in Z$ or 
there is $i\in\{h+1,\ldots,s\}$ such that $q_i=true$ and $\dist_{F'}(v_i,v_j)=2$ for $v_j\in Z\setminus\{v_i\}$;
 \end{itemize}
and $\omega(p,q,I,p_{h+1},q_{h+2},\ldots,p_s,q_s)=+\infty$ in all other cases.

We claim that 
$$w(H,p,q)=\min\omega(p,q,I,p_{h+1},q_{h+1},\ldots,p_s,q_s),$$
where the minimum is taken over all $I\subseteq\{1,\ldots,h\}$ and $p_i,q_i\in\{true,false\}$ for $i\in\{h+1,\ldots,s\}$.

First, we show that $w(H,p,q)\geq \min\omega(p,q,I,p_{h+1},q_{h+2},\ldots,p_s,q_s)$.
If $w(H,p,q)=+\infty$, then the inequality trivially holds. Let $w(H,p,q)<+\infty$.
Then there is a set $W\subseteq V(H)$  om minimum size such that 
\begin{itemize}
\item[i)] $W$ resolves $V(H)$ in $H'$,
\item[ii)] $H$ has a vertex $x$ such that $\dist_{H'}(x,v)=1$ for every $v\in W$ if and only if $p=true$, and 
\item[iii)] $H$ has a vertex $x$ such that $\dist_{H'}(x,v)= 2$ for every $v\in W$ if and only if $q=true$.
\end{itemize}
By the definition, $w(H,p,q)=|W|$. Let $W_i=W\cap X_i$ for $i\in\{1,\ldots,s\}$. 
Let $I=\{i|i\in\{1,\ldots,h\}, W_i\neq\emptyset\}$. Notice that $W_i\neq\emptyset$ for $i\in\{h+1,\ldots,s\}$ by Lemma~\ref{lem:mw}.
For $i\in\{h+1,\ldots,s\}$, let $p_i=true$ if there is a vertex $x\in X_i$ such that $\dist_{H'}(x,u)=1$ for $u\in W_i$, and 
let  $q_i=true$ if there is a vertex $x\in X_i$ such that $\dist_{H'}(x,u)=2$ for $u\in W_i$.

By Lemma~\ref{lem:mw}, $W_i$ resolves $X_i$ in the graph obtained from $H[X_i]$ by the addition of a universal vertex for $i\in\{h+1,\ldots,s\}$. Hence, $|W_i|\geq w(H[X_i],p_i,q_i)$ for $i\in\{h+1,\ldots,s\}$ and, therefore, $|W|\geq |I|+\sum_{i=h+1}^s w(H[X_i],p_i,q_i)$.

We show that a)--g) are fulfilled for $I$ and the defined values of $p_i,q_i$.

To show a), consider  distinct vertices $v_i,v_j$ of $F$. If $v_i\in Z$ or $v_j\in Z$, then it is straightforward to see that $Z$ resolves $v_i$ and $v_j$. Suppose that $i,j\in\{1,\ldots,h\}\setminus I$. 
Then $X_i,X_j$ are trivial modules with the unique vertices $x$ and $y$ respectively. Because $W$ resolves $V(H)$, there is $u\in W$ such that $\dist_{H'}(u,x)\neq\dist_{H'}(u,y)$. Consider the set $W_r$ containing $u$. It remains to observe that $v_r$ resolves $v_i$ and $v_j$, because
$\dist_{F'}(v_r,v_i)=\dist_{H'}(u,x)\neq\dist_{H'}(u,y)=\dist_{F'}(v_r,v_j)$. 
 
To prove b), assume that  $p_i=true$ for some $i\in\{h+1,\ldots,s\}$ and consider $j\in\{1,\ldots,h\}\setminus I$. Suppose that $\dist_{F'}(v_i,v_j)\neq 2$, i.e.,  $\dist_{F'}(v_i,v_j)=1$.
Then $X_i$ has a vertex $x$ adjacent to all the vertices of $W_i$. Let $y$ be the unique vertex of $X_j$. The set $W$ resolves $x,y$ and, therefore, there is $u\in W$ such that 
$\dist_{H'}(u,x)\neq \dist_{H'}(u,y)$.  If $u\in X_i$, then we have that $\dist_{H'}(u,x)=1=\dist_{F'}(v_i,v_j)=\dist_{H'}(u,y)$; a contradiction. 
Hence, $u\notin X_i$.  Let $X_r$ be the module containing $u$. Then we have that $\dist_{F'}(v_r,v_i)=\dist_{H'}(u,x)\neq\dist_{H'}(u,y)=\dist_{F'}(v_r,v_i)$.

Similarly, to obtain c), assume that $q_i=true$ for some $i\in\{h+1,\ldots,s\}$ and consider $j\in\{1,\ldots,h\}\setminus I$. Suppose that $\dist_{F'}(v_i,v_j)\neq 1$, i.e.,  $\dist_{F'}(v_i,v_j)=2$.
Then $X_i$ has a vertex $x$ at distance 2 from all the vertices of $W_i$. Let $y$ be the unique vertex of $X_j$. The set $W$ resolves $x,y$ and, therefore, there is $u\in W$ such that 
$\dist_{H'}(u,x)\neq \dist_{H'}(u,y)$.  If $u\in X_i$, then we have that $\dist_{H'}(u,x)=2=\dist_{F'}(v_i,v_j)=\dist_{H'}(u,y)$; a contradiction. 
Hence, $u\notin X_i$.  Let $X_r$ be the module containing $u$. Then we have that $\dist_{F'}(v_r,v_i)=\dist_{H'}(u,x)\neq\dist_{H'}(u,y)=\dist_{F'}(v_r,v_i)$.

To show d), suppose that $p_i=p_j=true$ for some distinct $i,j\in \{h+1,\ldots,s\}$ and assume that $\dist_{F'}(v_i,v_j)\neq 2$, i.e., $\dist_{F'}(v_i,v_j)=1$. Then $X_i$ has a vertex $x$ adjacent to all the vertices of $W_i$ and $X_j$ has a vertex $y$ adjacent to all the vertices of $W_j$. The set $W$ resolves $x,y$ and, therefore, there is $u\in W$ such that 
$\dist_{H'}(u,x)\neq \dist_{H'}(u,y)$.  If $u\in X_i$, then we have that $\dist_{H'}(u,x)=\dist_{F'}(v_i,v_j)=\dist_{H'}(u,y)$; a contradiction. Hence, $u\notin X_i$. By the same arguments, $u\notin X_j$. Let $X_r$ be the module containing $u$. Then we have that $\dist_{F'}(v_r,v_i)=\dist_{H'}(u,x)\neq\dist_{H'}(u,y)=\dist_{F'}(v_r,v_i)$.

To prove e), suppose that $q_i=q_j=true$ for some distinct $i,j\in \{h+1,\ldots,s\}$ and assume that $\dist_{F'}(v_i,v_j)\neq 1$, i.e., $\dist_{F'}(v_i,v_j)=2$. Then $X_i$ has a vertex $x$ at distance 2 to all the vertices of $W_i$ and $X_j$ has a vertex $y$ at distance 2 to all the vertices of $W_j$. The set $W$ resolves $x,y$ and, therefore, there is $u\in W$ such that 
$\dist_{H'}(u,x)\neq \dist_{H'}(u,y)$.  If $u\in X_i$, then we have that $\dist_{H'}(u,x)=\dist_{F'}(v_i,v_j)=\dist_{H'}(u,y)$; a contradiction. Hence, $u\notin X_i$. By the same arguments, $u\notin X_j$. Let $X_r$ be the module containing $u$. Then we have that $\dist_{F'}(v_r,v_i)=\dist_{H'}(u,x)\neq\dist_{H'}(u,y)=\dist_{F'}(v_r,v_i)$.

To see f), recall that $p=true$ if and only if $V(H)$ has a vertex $x$ that is adjacent to all the vertices of $W$. 
Suppose that $V(H)$ has a vertex $x$ that is adjacent to all the vertices of $W$. 
If $x\in X_i$ for $i\in\{1,\ldots,h\}\setminus I$, then 
$\dist_{F'}(v_i,v_j)=1$ for $v_j\in Z$. If $x\in X_i$ for $i\in\{h+1,\ldots,s\}$, then  $p_i=true$ and $\dist_{F'}(v_i,v_j)=1$ for $v_j\in Z\setminus\{v_i\}$.
Suppose that there is $i\in\{1,\ldots,h\}\setminus I$ such that $\dist_{F'}(v_i,v_j)=1$ for $v_j\in Z$ or 
there is $i\in\{h+1,\ldots,s\}$ such that $p_i=true$ and $\dist_{F'}(v_i,v_j)=1$ for $v_j\in Z\setminus\{v_i\}$. If 
there is $i\in\{1,\ldots,h\}\setminus I$ such that $\dist_{F'}(v_i,v_j)=1$ for $v_j\in Z$, then the unique vertex $x$ of $X_i$ is at distance 1 from all the vertices of $W$ and $p=true$.
If there is $i\in\{h+1,\ldots,s\}$ such that $p_i=true$, then there is $x\in X_i$ at distance 1 from each vertex of $W_i$.  If $\dist_{F'}(v_i,v_j)=1$ for $v_j\in Z\setminus\{v_i\}$, then $x$ at distance 1 from the vertices $W\setminus W_i$ and, therefore, $p=true$.

Similarly, to prove g),  recall that $q=true$ if and only if $V(H)$ has a vertex $x$ that is at distance 2 from every vertex of $W$. 
Suppose that $V(H)$ has a vertex $x$ that is at distance 2 from all the vertices of $W$. 
If $x\in X_i$ for $i\in\{1,\ldots,h\}\setminus I$, then 
$\dist_{F'}(v_i,v_j)=2$ for $v_j\in Z$. If $x\in X_i$ for $i\in\{h+1,\ldots,s\}$, then  $q_i=true$ and $\dist_{F'}(v_i,v_j)=2$ for $v_j\in Z\setminus\{v_i\}$.
Suppose that there is $i\in\{1,\ldots,h\}\setminus I$ such that $\dist_{F'}(v_i,v_j)=2$ for $v_j\in Z$ or 
there is $i\in\{h+1,\ldots,s\}$ such that $q_i=true$ and $\dist_{F'}(v_i,v_j)=2$ for $v_j\in Z\setminus\{v_i\}$. If 
there is $i\in\{1,\ldots,h\}\setminus I$ such that $\dist_{F'}(v_i,v_j)=2$ for $v_j\in Z$, then the unique vertex $x$ of $X_i$ is at distance 2 from all the vertices of $W$ and $q=true$.
If there is $i\in\{h+1,\ldots,s\}$ such that $q_i=true$, then there is $x\in X_i$ at distance 2 from each vertex of $W_i$.  If $\dist_{F'}(v_i,v_j)=2$ for $v_j\in Z\setminus\{v_i\}$, then $x$ at distance 2 from the vertices $W\setminus W_i$ and, therefore, $q=true$.

Because a)--g) are fulfilled, $w(H,p,q)\geq |W|\geq |I|+\sum_{i=h+1}^s w(H[X_i],p_i,q_i)=\omega(p,q,I,p_{h+1},q_{h+2},\ldots,p_s,q_s)$ and the claim follows.

Now we show that $w(H,p,q)\leq \min\omega(p,q,I,p_{h+1},q_{h+2},\ldots,p_s,q_s)$. Assume that $I$ and the values of $p_{h+1},q_{h+1},\ldots,p_s,q_s$ are chosen in such a way that
$\omega(p,q,I,p_{h+1},q_{h+2},\ldots,p_s,q_s)$ has the minimum possible value. If  $\omega(p,q,I,p_{h+1},q_{h+1},\ldots,p_s,q_s)=+\infty$, then, trivially, we have that
 $w(H,p,q)\leq \omega(p,q,I,p_{h+1},q_{h+1},\ldots,p_s,q_s)$. Suppose that $\omega(p,q,I,p_{h+1},q_{h+1},\ldots,p_s,q_s)<+\infty$. Then
$\omega(p,q,I,p_{h+1},q_{h+1},\ldots,p_s,q_s)=|I|+\sum_{i=h+1}^s w(H[X_i],p_i,q_i)$ and a)--g) are fulfilled for $p,q$, $I$ and the values of $p_{h+1},q_{h+1},\ldots,p_s,q_s$.

Notice that $w(H[X_i],p_i,q_i)<+\infty$ for $i\in\{h+1,\ldots,s\}$. For $i\in\{h+1,\ldots,s\}$, let 
$W_i\subseteq X_i$ be a set om minimum size such that 
\begin{itemize}
\item[i)] $W_i$ resolves $X_i$ in the graph $H'_i$ obtained from $H[X_i]$ by the addition of a universal vertex,
\item[ii)] $X_i$ has a vertex $x$ such that $\dist_{H'_i}(x,v)=1$ for every $v\in W_i$ if and only if $p_i=true$, and 
\item[iii)] $X_i$ has a vertex $x$ such that $\dist_{H_i'}(x,v)= 2$ for every $v\in W_i$ if and only if $q_i=true$.
\end{itemize}
By the definition, $w(H[X_i],p_i,q_i)=|W_i|$ for $i\in\{h+1,\ldots,s\}$. Let
$$W=(\cup_{i\in I}X_i)\cup(\cup_{i=h+1}^sW_i).$$
We have that $|W|=\omega(p,q,I,p_{h+1},q_{h+2},\ldots,p_s,q_s)$.

We claim $W$ is a resolving set for $V(H)$ in $H'$. 

Let $x,y$ be distinct vertices of $H$. We show that there is a vertex $u$ in $W$ that resolves $x$ and $y$ in $H'$. Clearly, it is sufficient to prove it for $x,y\in V(H)\setminus W$.
Let $X_i$ and $X_j$ be the modules that contain $x$ and $y$ respectively. If $i=j$, then a vertex $u\in W_i$ resolves $x$ and $y$ in $H_i'$ and, therefore, $u$ resolves $x$ and $y$ in $H'$. Suppose that $i\neq j$. 

Assume first that $i,j\in \{1,\ldots,h\}$. Then $i,j\in\{1,\ldots,h\}\setminus I$, because $X_1,\ldots,X_h$ are trivial. By a),  $Z$ resolves $V(F)$ in $F'$. Hence, there
is $v_r\in Z$ such that $\dist_{F'}(v_r,v_i)\neq \dist_{F'}(v_r,v_j)$. Notice that $W_r\neq\emptyset$ by the definition of $W_r$ and $Z$. Let $u\in W_r$. We have
that $\dist_{H'}(u,x)=\dist_{F'}(v_r,v_i)\neq\dist_{F'}(v_r,v_j)=\dist_{H'}(u,y)$. 

Let now  $i\in\{h+1,\ldots,s\}$ and $j\in\{1,\ldots,h\}$. If there are $u_1,u_2\in X_i$ such that $\dist_{H_i'}(u_1,x)\neq \dist_{H_i'}(u_2,x)$, then $u_1$ or $u_2$ resolves $x$ and $y$, because $\dist_{H'}(u_1,y)=\dist_{H'}(u_2,y)$. Assume that all the vertices of $W_i$ are at the same distance from $x$ in $H_i'$. 
Let $u\in W_i$. If $\dist_{H_i'}(u,x)=1$, then $p_i=true$ and, by b), $\dist_{F'}(v_i,v_j)=2$ or there is $v_r\in Z$ such that $r\neq i,j$ and $\dist_{F'}(v_r,v_i)\neq \dist_{F'}(v_r,v_j)$. If 
$\dist_{F'}(v_i,v_j)=2$, then $u$ resolves $x$ and $y$, as $\dist_{H'}(u,y)=2$. Otherwise, any vertex $u'\in W_r$ resolves $x$ and $y$.  Similarly,
if  $\dist_{H_i'}(u,x)=2$, then $q_i=true$ and, by c), $\dist_{F'}(v_i,v_j)=1$ or there is $v_r\in Z$ such that $r\neq i,j$ and $\dist_{F'}(v_r,v_i)\neq \dist_{F'}(v_r,v_j)$. If 
$\dist_{F'}(v_i,v_j)=1$, then $u$ resolves $x$ and $y$, as $\dist_{H'}(u,y)=1$. Otherwise, any vertex $u'\in W_r$ resolves $x$ and $y$. 

Finally, let $i,j\in\{h+1,\ldots,s\}$. If there are $u_1,u_2\in X_i$ such that $\dist_{H_i'}(u_1,x)\neq \dist_{H_i'}(u_2,x)$, then $u_1$ or $u_2$ resolves $x$ and $y$, because $\dist_{H'}(u_1,y)=\dist_{H'}(u_2,y)$. By the same arguments, if there are $u_1,u_2\in X_j$ such that $\dist_{H_j'}(u_1,y)\neq \dist_{H_i'}(u_2,y)$, then $u_1$ or $u_2$ resolves $x$ and $y$.
 Assume that all the vertices of $W_i$ are at the same distance from $x$ in $H_i'$ and all the vertices of $W_j$ are at the same distance from $y$ in $H_j'$. Let 
$u_1\in W_1$ and $u_2\in W_j$. If $\dist_{H_i'}(u_1,x)\neq\dist_{H_j'}(u_2,y)$, then  $u_1$ or $u_2$ resolves $x$ and $y$, because $\dist_{H'}(u_1,y)=\dist_{H'}(u_2,x)$.
Suppose that $\dist_{H_i'}(u_1,x)=\dist_{H_j'}(u_2,y)=1$. Then $p_i=p_j=true$ and, by d), $\dist_{F'}(v_i,v_j)=2$ or there is $v_r\in Z$ such that $r\neq i,j$ and $\dist_{F'}(v_r,v_i)\neq \dist_{F'}(v_r,v_j)$. If  $\dist_{F'}(v_i,v_j)=2$, then $u_1$ resolves $x$ and $y$. Otherwise, any vertex $u'\in W_r$ resolves $x$ and $y$. If 
 $\dist_{H_i'}(u_1,x)=\dist_{H_j'}(u_2,y)=2$, then $q_i=q_j=true$ and, by d), $\dist_{F'}(v_i,v_j)=1$ or there is $v_r\in Z$ such that $r\neq i,j$ and $\dist_{F'}(v_r,v_i)\neq \dist_{F'}(v_r,v_j)$. If  $\dist_{F'}(v_i,v_j)=1$, then $u_1$ resolves $x$ and $y$. Otherwise, any vertex $u'\in W_r$ resolves $x$ and $y$.

By f), $p=true$ if and only if there is $i\in\{1,\ldots,h\}\setminus I$ such that $\dist_{F'}(v_i,v_j)=1$ for $v_j\in Z$ or 
there is $i\in\{h+1,\ldots,s\}$ such that $p_i=true$ and $\dist_{F'}(v_i,v_j)=1$ for $v_j\in Z\setminus\{v_i\}$. 
If  there is $i\in\{1,\ldots,h\}\setminus I$ such that $\dist_{F'}(v_i,v_j)=1$, then the unique vertex $x\in X_i$ is at distance 1 from any vertex of $W$. If  there is $i\in\{h+1,\ldots,s\}$ such that $p_i=true$ and $\dist_{F'}(v_i,v_j)=1$ for $v_j\in Z\setminus\{v_i\}$, then there is a vertex $x\in X_i$ at distance 1 from each vertex of $W_i$, because $p_i=true$, and as  $\dist_{F'}(v_i,v_j)=1$ for $v_j\in Z\setminus\{v_i\}$, $x$ is at distance 1 from any  vertex of $W\setminus W_i$. Suppose that there is a vertex $x\in V(H)$ at distance 1 from each vertex of $W$. Let $X_i$ be the module containing $x$. If $i\in\{1,\ldots,h\}$, then $i\in\{1,\ldots,h\}\setminus I$ and $\dist_{F'}(v_i,v_j)=1$ for $v_j\in Z$. Hence, $p=true$. If $i\in\{h+1,\ldots,s\}$, then $p_i=true$, because $x$ is at distance 1 from the vertices of $W_i$. Because $x$ is at distance 1 from the vertices of $W\setminus W_i$, $\dist_{F'}(v_i,v_j)=1$ for $v_j\in Z\setminus\{v_i\}$. Therefore, $p=true$.

Similarly, by g), $q=true$ if and only if there is $i\in\{1,\ldots,h\}\setminus I$ such that $\dist_{F'}(v_i,v_j)=2$ for $v_j\in Z$ or 
there is $i\in\{h+1,\ldots,s\}$ such that $q_i=true$ and $\dist_{F'}(v_i,v_j)=2$ for $v_j\in Z\setminus\{v_i\}$. 
If  there is $i\in\{1,\ldots,h\}\setminus I$ such that $\dist_{F'}(v_i,v_j)=2$, then the unique vertex $x\in X_i$ is at distance 2 from any vertex of $W$. If  there is $i\in\{h+1,\ldots,s\}$ such that $q_i=true$ and $\dist_{F'}(v_i,v_j)=2$ for $v_j\in Z\setminus\{v_i\}$, then there is a vertex $x\in X_i$ at distance 2 from each vertex of $W_i$, because $q_i=true$, and as  $\dist_{F'}(v_i,v_j)=2$ for $v_j\in Z\setminus\{v_i\}$, $x$ is at distance 2 from any  vertex of $W\setminus W_i$. Suppose that there is a vertex $x\in V(H)$ at distance 2 from each vertex of $W$. Let $X_i$ be the module containing $x$. If $i\in\{1,\ldots,h\}$, then $i\in\{1,\ldots,h\}\setminus I$ and $\dist_{F'}(v_i,v_j)=2$ for $v_j\in Z$. Hence, $q=true$. If $i\in\{h+1,\ldots,s\}$, then $q_i=true$, because $x$ is at distance 2 from the vertices of $W_i$. Because $x$ is at distance 2 from the vertices of $W\setminus W_i$, $\dist_{F'}(v_i,v_j)=2$ for $v_j\in Z\setminus\{v_i\}$. Therefore, $q=true$.

We conclude that 
\begin{itemize}
\item[i)] $W$ resolves $V(H)$ in $H'$,
\item[ii)] $H$ has a vertex $x$ such that $\dist_{H'}(x,v)=1$ for every $v\in W$ if and only if $p=true$, and 
\item[iii)] $H$ has a vertex $x$ such that $\dist_{H'}(x,v)= 2$ for every $v\in W$ if and only if $q=true$.
\end{itemize}
Therefore, $w(H,p,q)\leq |W|=\omega(p,q,I,p_{h+1},q_{h+2},\ldots,p_s,q_s)$. 

This concludes Case~3. 

Our next aim is to explain how to compute $\md(G)$. Recall that $G$ is a connected graph of modular-width at most $t$. Hence, $G$ is either a single-vertex graph, or is a join of two graphs $H_1$ and $H_2$ of modular-width at most $t$, or  $V(G)$ can be partitioned into $s\leq t$ modules $X_1,\ldots,X_s$ such that $\mw(G[X_i])\leq t$ for $i\in\{1,\ldots,s\}$.

\medskip
\noindent
{\bf Case~1.} $|V(G)|=1$. It is straightforward to see that $\md(G)=1$.

\medskip
\noindent
{\bf Case~2.} $G$ is a join of $H_1$ and $H_2$. 
Assume without loss of generality that $|V(H_1)|\leq |V(H_2)|$. 

If $|V(H_1)|=|V(H_2)|=1$, then $\md(G)=1$. 

Suppose that $|V(H_1)|=1$, $|V(H_2)|\geq 2$ and the values of $w(H_2,p,q)$ are already computed for $p,q\in\{true,false\}$. Clearly, the single vertex of $H_1$ is at distance 1 from any vertex of $H_2$ in $G$, and this single vertex is in a resolving set or not.
Then by Lemma~\ref{lem:mw},
$\md(G)=\min\{w(H_2,false,true), w(H_2,false,false), w(H_2,true,true)+1,\linebreak w(H_2,true,false)+1\}$.

Suppose that $|V(H_1)|,|V(H_2)|\geq 2$ and the values of $w(H_i,p,q)$ are already computed for $i\in\{1,2\}$ and $p,q\in\{true,false\}$. Notice that for $x\in V(H_1)$ and $y\in V(H_2)$, $\dist_{G}(x,y)=1$, and any resolving set has at least one vertex in $H_1$ and at least one vertex in $H_2$.
Then by Lemma~\ref{lem:mw},
$\md(G)=\min\{w(H_1,p_1,q_1)+w(H_2,p_2,q_2)\mid p_i,q_i\in\{true,false\}\text{ for }i\in\{1,2\}\text{ and }\{p_1,p_2\}\neq\{true,true\}\}$.

\medskip
\noindent
{\bf Case~3.} $V(G)$ is partitioned into $s\leq t$ non-empty modules $X_1,\ldots,X_s$, $s\geq 2$. 
Again, Case~2 can be seen as a special case of Case~3, but we keep Case~2 to make the description of computing $\md(G)$ more clear.
We assume that $X_1,\ldots,X_h$ are \emph{trivial}, i.e, $|X_i|=1$ for $i\in\{1,\ldots,h\}$; it can happen that $h=0$. 
Consider the \emph{prime} graph $F$ with a vertex set $\{v_1,\ldots,v_s\}$ such that $v_i$ is adjacent to $v_j$ if and only if the vertices of $X_i$ are adjacent to the vertices of $X_j$ for distinct $i,j\in\{1,\ldots,s\}$. Observe that if $x\in X_i$ and $y\in X_j$ for distinct $i,j\in\{1,\ldots,s\}$, then $\dist_{G}(x,y)=\dist_{F}(v_i,v_j)$.

For a set of indices $I\subseteq\{1,\ldots,h\}$ and boolean variables $p_i,q_i$, where $i\in\{h+1,\ldots,s\}$, we define 
$$\omega(I,p_{h+1},q_{h+2},\ldots,p_s,q_s)=|I|+\sum_{i=h+1}^s w(G[X_i],p_i,q_i)$$ 
if the following holds:
\begin{itemize}
\item[a)] the set $Z=\{v_i\mid i\in I\cup\{h+1,\ldots,s\}\}$ is a resolving set for $F$,
\item[b)] if $p_i=true$ for some $i\in\{h+1,\ldots,s\}$, then for each $j\in\{1,\ldots,h\}\setminus I$, $\dist_{F}(v_i,v_j)\geq2$ or there is $v_r\in Z$ such that $r\neq i,j$ and $\dist_{F}(v_r,v_i)\neq \dist_{F'}(v_r,v_j)$,
\item[c)] if $q_i=true$ for some $i\in\{h+1,\ldots,s\}$, then for each $j\in\{1,\ldots,h\}\setminus I$, $\dist_{F}(v_i,v_j)\neq 2$ or there is $v_r\in Z$ such that $r\neq i,j$ and $\dist_{F}(v_r,v_i)\neq \dist_{F'}(v_r,v_j)$,
\item[d)] if $p_i=p_j=true$ for some distinct $i,j\in \{h+1,\ldots,s\}$, then $\dist_{F}(v_i,v_j)\geq 2$ or there is $v_r\in Z$ such that $r\neq i,j$ and $\dist_{F'}(v_r,v_i)\neq \dist_{F}(v_r,v_j)$,
\item[e)] if $q_i=q_j=true$ for some distinct $i,j\in \{h+1,\ldots,s\}$, then $\dist_{F}(v_i,v_j)\neq 2$ or there is $v_r\in Z$ such that $r\neq i,j$ and $\dist_{F'}(v_r,v_i)\neq \dist_{F}(v_r,v_j)$,
 \end{itemize}
and $\omega(I,p_{h+1},q_{h+2},\ldots,p_s,q_s)=+\infty$ in all other cases.

We claim that 
$$\md(G)=\min\omega(I,p_{h+1},q_{h+1},\ldots,p_s,q_s),$$
where the minimum is taken over all $I\subseteq\{1,\ldots,h\}$ and $p_i,q_i\in\{true,false\}$ for $i\in\{h+1,\ldots,s\}$.
 
First, we show that $\md(G)\geq \min\omega(p,q,I,p_{h+1},q_{h+2},\ldots,p_s,q_s)$.

Let $W\subseteq V(G)$ be a resolving set om minimum size. Clearly, $\md(G)=|W|$.
Let $W_i=W\cap X_i$ for $i\in\{1,\ldots,s\}$. 
Let $I=\{i|i\in\{1,\ldots,h\}, W_i\neq\emptyset\}$. Notice that $W_i\neq\emptyset$ for $i\in\{h+1,\ldots,s\}$ by Lemma~\ref{lem:mw}.
For $i\in\{h+1,\ldots,s\}$, let $p_i=true$ if there is a vertex $x\in X_i$ such that $\dist_G(x,u)=1$ for $u\in W_i$, and 
let  $q_i=true$ if there is a vertex $x\in X_i$ such that $\dist_G(x,u)=2$ for $u\in W_i$.

By Lemma~\ref{lem:mw}, $W_i$ resolves $X_i$ in the graph obtained from $G[X_i]$ by the addition of a universal vertex for $i\in\{h+1,\ldots,s\}$. Hence, $|W_i|\geq w(G[X_i],p_i,q_i)$ for $i\in\{h+1,\ldots,s\}$ and, therefore, $|W|\geq |I|+\sum_{i=h+1}^s w(G[X_i],p_i,q_i)$.

We show that a)--e) are fulfilled for $I$ and the defined values of $p_i,q_i$.

To show a), consider  distinct vertices $v_i,v_j$ of $F$. If $v_i\in Z$ or $v_j\in Z$, then it is straightforward to see that $Z$ resolves $v_i$ and $v_j$. Suppose that $i,j\in\{1,\ldots,h\}\setminus I$. 
Then $X_i,X_j$ are trivial modules with the unique vertices $x$ and $y$ respectively. Because $W$ is a resolving set for $G$, there is $u\in W$ such that $\dist_{G}(u,x)\neq\dist_{G}(u,y)$. Consider the set $W_r$ containing $u$. It remains to observe that $v_r$ resolves $v_i$ and $v_j$, because
$\dist_{F}(v_r,v_i)=\dist_{G}(u,x)\neq\dist_{G}(u,y)=\dist_{F}(v_r,v_j)$. 
 
To prove b), assume that  $p_i=true$ for some $i\in\{h+1,\ldots,s\}$ and consider $j\in\{1,\ldots,h\}\setminus I$. Suppose that $\dist_{F}(v_i,v_j)=1$.
Then $X_i$ has a vertex $x$ adjacent to all the vertices of $W_i$. Let $y$ be the unique vertex of $X_j$. The set $W$ resolves $x,y$ and, therefore, there is $u\in W$ such that 
$\dist_{G}(u,x)\neq \dist_{G}(u,y)$.  If $u\in X_i$, then we have that $\dist_{G}(u,x)=1=\dist_{F}(v_i,v_j)=\dist_{G}(u,y)$; a contradiction. 
Hence, $u\notin X_i$.  Let $X_r$ be the module containing $u$. Then we have that $\dist_F{v_r,v_i}=\dist_{G}(u,x)\neq\dist_{G}(u,y)=\dist_F{v_r,v_i}$.

Similarly, to obtain c), assume that $q_i=true$ for some $i\in\{h+1,\ldots,s\}$ and consider $j\in\{1,\ldots,h\}\setminus I$. Suppose that $\dist_{F}(v_i,v_j)=2$.
Then $X_i$ has a vertex $x$ at distance 2 from all the vertices of $W_i$. §
Let $y$ be the unique vertex of $X_j$. The set $W$ resolves $x,y$ and, therefore, there is $u\in W$ such that 
$\dist_{G}(u,x)\neq \dist_{G}(u,y)$.  If $u\in X_i$, then we have that $\dist_{G}(u,x)=2=\dist_{F}(v_i,v_j)=\dist_{G}(u,y)$; a contradiction. 
Hence, $u\notin X_i$.  Let $X_r$ be the module containing $u$. Then we have that $\dist_F(v_r,v_i)=\dist_{G}(u,x)\neq\dist_{G}(u,y)=\dist_F(v_r,v_i)$.

To show d), suppose that $p_i=p_j=true$ for some distinct $i,j\in \{h+1,\ldots,s\}$ and assume that $\dist_{F}(v_i,v_j)=1$. Then $X_i$ has a vertex $x$ adjacent to all the vertices of $W_i$ and $X_j$ has a vertex $y$ adjacent to all the vertices of $W_j$. The set $W$ resolves $x,y$ and, therefore, there is $u\in W$ such that 
$\dist_{G}(u,x)\neq \dist_{G}(u,y)$.  If $u\in X_i$, then we have that $\dist_{G}(u,x)=\dist_{F}(v_i,v_j)=\dist_{G}(u,y)$; a contradiction. Hence, $u\notin X_i$. By the same arguments, $u\notin X_j$. Let $X_r$ be the module containing $u$. Then we have that $\dist_F(v_r,v_i)=\dist_{G}(u,x)\neq\dist_{G}(u,y)=\dist_F(v_r,v_i)$.

To prove e), suppose that $q_i=q_j=true$ for some distinct $i,j\in \{h+1,\ldots,s\}$ and assume that $\dist_{F'}(v_i,v_j)=2$. Then $X_i$ has a vertex $x$ at distance 2 from all the vertices of $W_i$ and $X_j$ has a vertex $y$ at distance 2 from all the vertices of $W_j$. The set $W$ resolves $x,y$ and, therefore, there is $u\in W$ such that 
$\dist_{G}(u,x)\neq \dist_{G}(u,y)$.  If $u\in X_i$, then we have that $\dist_{G}(u,x)=\dist_{F}(v_i,v_j)=\dist_{G}(u,y)$; a contradiction. Hence, $u\notin X_i$. By the same arguments, $u\notin X_j$. Let $X_r$ be the module containing $u$. Then we have that $\dist_F(v_r,v_i)=\dist_{G}(u,x)\neq\dist_{G}(u,y)=\dist_F(v_r,v_i)$.
 
Because a)--e) are fulfilled, $\md(G)\geq |W|\geq |I|+\sum_{i=h+1}^s w(G[X_i],p_i,q_i)=\omega(I,p_{h+1},q_{h+2},\ldots,p_s,q_s)$ and the claim follows.

Now we show that $\md(G)\leq \min\omega(I,p_{h+1},q_{h+2},\ldots,p_s,q_s)$. Assume that $I$ and the values of $p_{h+1},q_{h+1},\ldots,p_s,q_s$ are chosen in such a way that
$\omega(I,p_{h+1},q_{h+2},\ldots,p_s,q_s)$ has the minimum possible value. If  \linebreak $\omega(I,p_{h+1},q_{h+1},\ldots,p_s,q_s)=+\infty$, then, trivially, we have that
 $\md(G)\leq \omega(I,p_{h+1},q_{h+1},\ldots,p_s,q_s)$. Suppose that $\omega(I,p_{h+1},q_{h+1},\ldots,p_s,q_s)<+\infty$. Then
$\omega(I,p_{h+1},q_{h+1},\ldots,p_s,q_s)=|I|+\sum_{i=h+1}^s w(H[X_i],p_i,q_i)$ and a)--e) are fulfilled for $I$ and the values of $p_{h+1},q_{h+1},\ldots,p_s,q_s$.

Notice that $w(G[X_i],p_i,q_i)<+\infty$ for $i\in\{h+1,\ldots,s\}$. For $i\in\{h+1,\ldots,s\}$, let 
$W_i\subseteq X_i$ be a set om minimum size such that 
\begin{itemize}
\item[i)] $W_i$ resolves $X_i$ in the graph $H'_i$ obtained from $G[X_i]$ by the addition of a universal vertex,
\item[ii)] $X_i$ has a vertex $x$ such that $\dist_{H'_i}(x,v)=1$ for every $v\in W_i$ if and only if $p_i=true$, and 
\item[iii)] $X_i$ has a vertex $x$ such that $\dist_{H_i'}(x,v)= 2$ for every $v\in W_i$ if and only if $q_i=true$.
\end{itemize}
By the definition, $w(G[X_i],p_i,q_i)=|W_i|$ for $i\in\{h+1,\ldots,s\}$. Let
$$W=(\cup_{i\in I}X_i)\cup(\cup_{i=h+1}^sW_i).$$
We have that $|W|=\omega(p,q,I,p_{h+1},q_{h+2},\ldots,p_s,q_s)$.

We claim $W$ is a resolving set for $G$. 

Let $x,y$ be distinct vertices of $G$. We show that there is a vertex $u$ in $W$ that resolves $x$ and $y$ in $G$. Clearly, it is sufficient to prove it for $x,y\in V(G)\setminus W$.
Let $X_i$ and $X_j$ be the modules that contain $x$ and $y$ respectively. If $i=j$, then a vertex $u\in W_i$ resolves $x$ and $y$ in $H_i'$ and, therefore, $u$ resolves $x$ and $y$ in $G$. Suppose that $i\neq j$. 

Assume first that $i,j\in \{1,\ldots,h\}$. Then $i,j\in\{1,\ldots,h\}\setminus I$, because $X_1,\ldots,X_h$ are trivial. By a),  $Z$ is a resolving set for $F$. Hence, there
is $v_r\in Z$ such that $\dist_{F}(v_r,v_i)\neq \dist_{F}(v_r,v_j)$. Notice that $W_r\neq\emptyset$ by the definition of $W_r$ and $Z$. Let $u\in W_r$. We have
that $\dist_{G}(u,x)=\dist_{F}(v_r,v_i)\neq\dist_{F}(v_r,v_j)=\dist_{G}(u,y)$. 

Let now  $i\in\{h+1,\ldots,s\}$ and $j\in\{1,\ldots,h\}$. If there are $u_1,u_2\in X_i$ such that $\dist_{H_i'}(u_1,x)\neq \dist_{H_i'}(u_2,x)$, then $u_1$ or $u_2$ resolves $x$ and $y$, because $\dist_{G}(u_1,y)=\dist_{G}(u_2,y)$. Assume that all the vertices of $W_i$ are at the same distance from $x$ in $H_i'$. 
Let $u\in W_i$. If $\dist_{H_i'}(u,x)=1$, then $p_i=true$ and, by b), $\dist_{F}(v_i,v_j)\geq 2$ or there is $v_r\in Z$ such that $r\neq i,j$ and $\dist_{F}(v_r,v_i)\neq \dist_{'}(v_r,v_j)$. If 
$\dist_{F'}(v_i,v_j)\geq 2$, then $u$ resolves $x$ and $y$, as $\dist_{G}(u,y)=2$. Otherwise, any vertex $u'\in W_r$ resolves $x$ and $y$.  Similarly,
if  $\dist_{H_i'}(u,x)=2$, then $q_i=true$ and, by c), $\dist_{F}(v_i,v_j)\neq 2$ or there is $v_r\in Z$ such that $r\neq i,j$ and $\dist_{F}(v_r,v_i)\neq \dist_{F}(v_r,v_j)$. If 
$\dist_{F'}(v_i,v_j)\neq 2$, then $u$ resolves $x$ and $y$, as $\dist_{G}(u,y)\neq 2$. Otherwise, any vertex $u'\in W_r$ resolves $x$ and $y$.

Finally, let $i,j\in\{h+1,\ldots,s\}$. If there are $u_1,u_2\in X_i$ such that $\dist_{H_i'}(u_1,x)\neq \dist_{H_i'}(u_2,x)$, then $u_1$ or $u_2$ resolves $x$ and $y$, because $\dist_{G}(u_1,y)=\dist_{F'}(u_2,y)$. By the same arguments, if there are $u_1,u_2\in X_j$ such that $\dist_{H_j'}(u_1,y)\neq \dist_{H_i'}(u_2,y)$, then $u_1$ or $u_2$ resolves $x$ and $y$.
 Assume that all the vertices of $W_i$ are at the same distance from $x$ in $H_i'$ and all the vertices of $W_j$ are at the same distance from $y$ in $H_j'$. Let 
$u_1\in W_1$ and $u_2\in W_j$. If $\dist_{H_i'}(u_1,x)\neq\dist_{H_j'}(u_2,y)$, then  $u_1$ or $u_2$ resolves $x$ and $y$, because $\dist_{G}(u_1,y)=\dist_{G}(u_2,x)$.
Suppose that $\dist_{H_i'}(u_1,x)=\dist_{H_j'}(u_2,y)=1$. Then $p_i=p_j=true$ and, by d), $\dist_{F'}(v_i,v_j)\geq2$ or there is $v_r\in Z$ such that $r\neq i,j$ and $\dist_{F}(v_r,v_i)\neq \dist_{F}(v_r,v_j)$. If  $\dist_{F}(v_i,v_j)\geq 2$, then $u_1$ resolves $x$ and $y$. Otherwise, any vertex $u'\in W_r$ resolves $x$ and $y$. If 
 $\dist_{H_i'}(u_1,x)=\dist_{H_j'}(u_2,y)=2$, then $q_i=q_j=true$ and, by d), $\dist_{F}(v_i,v_j)\neq 2$ or there is $v_r\in Z$ such that $r\neq i,j$ and $\dist_{F}(v_r,v_i)\neq \dist_{F}(v_r,v_j)$. If  $\dist_{F}(v_i,v_j)\neq2$, then $u_1$ resolves $x$ and $y$. Otherwise, any vertex $u'\in W_r$ resolves $x$ and $y$.

We have that $W$ is a resolving set for $G$ and, therefore,  $\md(G)\leq |W|=\omega(I,p_{h+1},q_{h+2},\ldots,p_s,q_s)$. 

\medskip
Recall that the modular-width of a graph can be computed in linear time by the algorithm of Tedder et al.~\cite{TedderCHP08}, and 
this algorithm outputs the algebraic expression of $G$
corresponding to the procedure of its construction from isolated vertices by  the disjoint union and join operation and decomposing $H$ into at most $t$ modules.
We construct such a decomposition and consider the rooted tree corresponding to the algebraic expression. We compute the values of $w(H,p,q)$ for the graphs $H$ corresponding to the internal nodes of the tree and then compute $\md(G)$ for the root corresponding to $G$. 

To evaluate the running time, observe that computing $w(H,p,q)$ for the disjoint union or join of two graphs demands $O(1)$ operations. To compute $w(H,p,q)$ in the case when $V(H)$ is partitioned into $s\leq t$ modules, we consider at most $4^t$ possibilities to choose $I$ and $p_i,q_i$ for $i\in\{h+1,\ldots,s\}$. Then the conditions a)--g) can be verified in time $O(t^3)$. Hence, the total time is $O(t^34^t)$. Similarly, the final computation of $\md(G)$ can be performed in time $O(t^34^t)$. We conclude that the running time is $O(t^34^t\cdot n)$ for a given  decomposition. Since the algorithm of Tedder et al.~\cite{TedderCHP08} is linear, we solve \textsc{Minimum Metric Dimension} in time $O(t^34^t\cdot n +m)$.
\end{proof}

\section{Conclusions}
We have essentially shown that {\mdim} can be solved in polynomial time on graphs of constant degree and tree-length. For this, amongst other things, we used the fact that such graphs have constant treewidth. Therefore, the most natural step forward would be to attempt to extend these results to graphs of constant treewidth which do not necessarily have bounded degree or tree-length. In fact, we point out that it is not known whether {\mdim} is polynomial-time solvable even on graphs of treewidth at most 2.

\end{document}